\documentclass[preprint,12pt,authoryear]{elsarticle}
\usepackage{booktabs}
\usepackage{array}
\usepackage{url}
\usepackage{bm}
\usepackage{algorithmicx}
\usepackage{algorithm,algpseudocode}
\usepackage{graphicx,subcaption}
\usepackage{hyperref}
\usepackage{geometry}
\usepackage{color}
\usepackage{xcolor}
\usepackage{cancel}

\usepackage{amssymb}
\usepackage{amsmath}

\usepackage{amsthm}

\usepackage{lipsum}
\makeatletter
\renewenvironment{proof}[1][\proofname]{\par
  \pushQED{\qed}%
  \normalfont \topsep6\p@\@plus6\p@\relax
  \list{}{\leftmargin=1em
          \rightmargin=\leftmargin
          \settowidth{\itemindent}{\itshape#1}%
          \labelwidth=\itemindent
          \parsep=0pt \listparindent=\parindent 
  }
  \item[\hskip\labelsep
        \itshape
    #1\@addpunct{.}]\ignorespaces
}{%
  \popQED\endlist\@endpefalse
}
\makeatother
\newcommand{\BlackBox}{\rule{1.5ex}{1.5ex}}  
\ifdefined\proof
    \renewenvironment{proof}{\par\noindent{\bf Proof\ }}{\hfill\BlackBox\\[2mm]}
\else
    \newenvironment{proof}{\par\noindent{\bf Proof\ }}{\hfill\BlackBox\\[2mm]}
\fi
\newtheorem{theorem}{Theorem}
\theoremstyle{remark}
\newtheorem*{remark}{Remark}
\usepackage{multirow}

\journal{Spatial Statistics}

\begin{document}

\begin{frontmatter}



\title{A Delayed Acceptance Auxiliary Variable MCMC for Spatial Models with Intractable Likelihood Function\tnoteref{t1}} 
\tnotetext[t1]{Jong Hyeon Lee and Jongmin Kim contributed equally to this work.}


\author[label1]{Jong Hyeon Lee} 
\author[label1]{Jongmin Kim} 
\author[label1]{Heesang Lee}
\author[label1,label2]{Jaewoo Park\corref{cor1}}%
\ead{jwpark88@yonsei.ac.kr}
\affiliation[label1]{organization={Department of Statistics and Data Science},
            addressline={Yonsei University}, 
            city={Seoul},
            country={South Korea}}
            
\affiliation[label2]{organization={Department of Applied Statistics},
            addressline={Yonsei University}, 
            city={Seoul},
            country={South Korea}}

\cortext[cor1]{Corresponding author}

\begin{abstract}
A large class of spatial models contains intractable normalizing functions, such as spatial lattice models, interaction spatial point processes, and social network models. Bayesian inference for such models is challenging since the resulting posterior distribution is doubly intractable. Although auxiliary variable MCMC (AVM) algorithms are known to be the most practical, they are computationally expensive due to the repeated auxiliary variable simulations. To address this, we propose delayed-acceptance AVM (DA-AVM) methods, which can reduce the number of auxiliary variable simulations. The first stage of the kernel uses a cheap surrogate to decide whether to accept or reject the proposed parameter value. The second stage guarantees detailed balance with respect to the posterior. The auxiliary variable simulation is performed only on the parameters accepted in the first stage. We construct various surrogates specifically tailored for doubly intractable problems, including subsampling strategy, Gaussian process emulation, and frequentist estimator-based approximation. We validate our method through simulated and real data applications, demonstrating its practicality for complex spatial models. 

\end{abstract}




\end{frontmatter}


\section{Introduction}
\label{sec:intro}

Intractable spatial models arise in many disciplines, for instance, Potts models \citep{potts1952some} for discrete lattice data, interaction point processes \citep{strauss1975model, attrep} for spatial point pattern data, and exponential random graph models (ERGMs) \citep{robins2007introduction} for social network data. Bayesian inference for such models is challenging because the likelihood functions involve intractable normalizing functions, which are functions of the parameters of interest. Let $\mathbf{x}\in \mathcal{X}$ be a realization from an unnormalized probability model $h(\mathbf{x}|\bm{\theta})$ with a model parameter $\bm{\theta} \in \bm{\Theta}$. The unnormalized probability model has an intractable normalizing function $Z(\bm{\theta}) = \int_{\mathcal{X}} h(\mathbf{x} | \bm{\theta}) d\mathbf{x}$. With a prior $p(\bm \theta)$ the posterior of $\bm{\theta}$ is 
\[
\pi(\bm{\theta}|\mathbf{x}) = \frac{p(\bm \theta) h(\mathbf{x}|\bm \theta)}{p(\mathbf{x})Z(\bm\theta)},
\]
where the marginal likelihood is 
\[
p(\mathbf{x})=\int_{\bm\Theta} \frac{p(\bm \theta)h(\mathbf{x}|\bm \theta)}{Z(\bm\theta)}d\bm\theta.
\]
Here, $p(\mathbf{x})$ is intractable, but can be ignored in the inference, since it does not depend on $\bm\theta$. In contrast, the other intractable term $Z(\bm\theta)$ is a function of $\bm\theta$ and therefore cannot be disregarded, which makes the application of standard Markov chain Monte Carlo (MCMC) algorithms challenging. \cite{Murray2006} referred to $\pi(\bm{\theta}|\mathbf{x})$ as a doubly-intractable distribution because the posterior involves both $p(\mathbf{x})$ and $Z(\bm\theta)$ as intractable terms.


To address this challenge, several Bayesian approaches have been developed, and \cite{park2018bayesian} classified them into two categories: (1) {\textit{likelihood approximation approaches}} and (2) {\textit{auxiliary variable approaches}}. The likelihood approximation approaches \citep{atchade2008bayesian, lyne2015russian, alquier2014noisy} directly approximate $Z(\bm{\theta})$ using importance sampling estimates and plug these estimates into the acceptance probability of the Metropolis-Hastings (MH) algorithm. In contrast, the auxiliary variable approaches \citep{moller2006, Murray2006, liang2010double, liang2015adaptive} simulate an auxiliary variable from $h(\cdot | \bm{\theta})$ at each iteration to cancel out the normalizing functions in the acceptance probability. As an extension, \cite{caimo2011bayesian} developed a population MCMC approach for ERGMs to improve the mixing of the standard AVM algorithm. 
Through an extensive numerical study, \cite{park2018bayesian} reported that auxiliary variable MCMC (AVM) methods \citep{Murray2006, liang2010double} are the most efficient in terms of effective sample size per time, defined as the effective sample size (accounting for autocorrelation among MCMC samples) divided by the wall-clock computing time. This measures the sampling efficiency of an algorithm relative to computational cost and is widely used to compare the practical performance of MCMC methods. Therefore, we also compare our proposed methods with these AVM methods. Due to its
ease of use, AVM approaches have been widely used in many applications. Examples include 
astrophysical problem \citep{tak2018repelling}, longitudinal item response model \citep{park2022bayesian}, and spatial count data exhibiting under- and overdispersion \citep{kang2024fast}. However, when the dimension of $\mathbf{x}$ becomes large, AVM methods become computationally expensive because auxiliary variable simulations from $h(\cdot|\bm{\theta})$ require the longer length of the Markov chain. 

In this manuscript, we propose delayed-acceptance AVM (DA-AVM) for intractable spatial models. The DA-MCMC method introduced by \cite{christen2005markov} is a two-stage Metropolis-Hastings (MH) algorithm that reduces the computational burden associated with calculating the likelihoods of complex models using the initial screening step. In the first stage of the kernel, a computationally cheap surrogate is used to evaluate the proposed parameters. If the proposal is accepted in the first stage, the algorithm evaluates the expensive likelihood function in the second stage. The final acceptance or rejection of the proposal is based on this correction step. Due to its efficiency and flexibility, the DA approaches have been widely used by constructing surrogates \citep{golightly2015delayed,sherlock2017adaptive,cao2024using} or partitioning large datasets \citep{banterle2015accelerating,quiroz2018speeding}. Our work is motivated by these recent computational approaches.

In the first stage of DA-AVM, we construct surrogates tailored for a wide variety of doubly intractable problems. Specifically, we investigate the subsampling strategy, Gaussian process emulation, and frequentist estimator-based approximation such as Monte Carlo maximum likelihood (MCML) \citep{geyer1992constrained} or maximum pseudo-likelihood (MPL) \citep{besag1974spatial}. From these surrogates, we can quickly rule out implausible regions of $\bm{\Theta}$ without simulating an auxiliary variable belonging to $\mathcal{X}$. If the proposal is accepted in the first stage, we simulate an auxiliary variable to decide the final acceptance. Since the algorithm satisfies the detailed balance condition, DA-AVM produces the same posterior distribution as the standard AVM while requiring fewer auxiliary variable simulations. Note that the performance of purely emulation-based approaches \citep{park2020function,vu2023warped} greatly depends on the accuracy of the surrogate model. \cite{park2020function} developed a two-stage MCMC algorithm that approximates $Z(\bm{\theta})$ via importance sampling and interpolates it with a Gaussian process emulator. Theorem 1 of \cite{park2020function} showed that, as the number of importance samples and design points increases, the Markov chain samples from the Gaussian process emulation algorithm will be close to $\pi(\bm{\theta}|\mathbf{x})$ in terms of total variation distance. In practice, however, constructing an accurate emulator is challenging in high-dimensional $\bm{\Theta}$: obtaining accurate importance sampling estimates requires importance parameters close to the MLE, and, for the Gaussian process emulator, the number of design points grows exponentially with the dimension. Although \cite{park2020function} proposed to address this issue by constructing the emulator with a short run of another algorithm (e.g., AVM), the approach is heuristic and becomes impractical beyond five dimensions. Our study also finds that as the parameter dimension increases (Figures 2–4), Gaussian process emulation fails to accurately approximate the true posterior. On the other hand, our DA-AVM is robust in the surrogate model construction because the second stage of the kernel corrects the discrepancy, ensuring convergence to the target posterior.

Approximate Bayesian computation (ABC) methods \citep{beaumont2002approximate,marin2012approximate} provide a natural framework for models with intractable likelihood functions. For example, \cite{vihrs2022approximate} developed a new point process model that captures both aggregation and repulsion, and employed an ABC method based on summary statistics specifically designed for the model. Motivated by the AVM methods, \cite{stoica2017abc} developed the ABC Shadow algorithm, which generates an auxiliary variable from $h(\cdot|\mathbf{x})$ to bypass the evaluation of the intractable $Z(\bm{\theta})$ and approximate the posterior distribution within the ABC framework. Extending this idea, \cite{stoica2021shadow} proposed the Shadow simulated annealing algorithm, which improves convergence by introducing a temperature schedule. \cite{laporte2022morpho} proposed a social process model that integrates spatial processes with ERGMs, using the ABC Shadow algorithm to approximate the intractable posterior. To reduce the number of synthetic data simulations in ABC methods, \cite{everitt2021delayed} incorporated a delayed acceptance strategy, which also motivates our work.

Compared to existing works, DA-AVM is an easy-to-use yet computationally efficient method. In contrast to ABC approaches \citep[cf.][]{shirota2017approximate, vihrs2022approximate}, it does not rely on summary statistics or tolerance levels that strongly affect the quality of the approximate posterior. Compared to likelihood-approximation approaches \citep{atchade2008bayesian, lyne2015russian, alquier2014noisy}, it requires tuning far fewer components and is fast because it does not require multiple importance samples at each MCMC iteration. Similar to AVM methods, our framework is easy to implement as long as auxiliary variables can be generated from the model, and the delayed acceptance step further improves efficiency by reducing the number of auxiliary variable simulations. Due to its generality, the framework is applicable to a wide range of models, including lattice spatial models, spatial point processes, network models, and time series models, as demonstrated in our numerical study.

The remainder of this manuscript is organized as follows. In Section~\ref{sec:2}, we introduce AVM algorithms for intractable spatial models and discuss their computational challenges. We also describe the background for DA-MCMC approaches. In Section~\ref{sec:3}, we propose an efficient DA-AVM with various surrogate candidates. We show that our DA-AVM satisfies the detailed balance condition with respect to the target posterior, and the resulting Markov chain is ergodic. In Section~\ref{sec:simulation}, we study the performance of our method with four intractable models—three spatial and one non-spatial—illustrating that our method can reduce computational costs. In Section~\ref{sec:discuss}, we summarize the key findings and contributions of this work.

\section{Computational Methods}
\label{sec:2}

\subsection{Auxiliary Variable MCMC}

As a seminal work, \cite{moller2006} developed the AVM method, which constructs the joint posterior of model parameters and an auxiliary variable to avoid direct evaluation of $Z(\bm{\theta})$. Let $f(\mathbf{y}|\bm{\theta},\mathbf{x})$ denote the conditional distribution of the auxiliary variable $\mathbf{y} \in \mathcal{X}$. Then the joint posterior distribution can be written as  
\begin{align}
\pi(\bm{\theta}, \mathbf{y} | \mathbf{x})
&= f(\mathbf{y} | \bm{\theta}, \mathbf{x}) \, \pi(\bm{\theta} | \mathbf{x}) \notag \\
&\propto f(\mathbf{y}|\bm{\theta},\mathbf{x})p(\bm \theta)\frac{h(\mathbf x|\bm \theta)}{Z(\bm \theta)}.
\label{jointpostAVM}
\end{align}
The MH algorithm is implemented using a joint proposal that factorizes into two components as follows:
\begin{equation}
q((\bm{\theta}, \mathbf{y}) \rightarrow (\bm{\theta}^*, \mathbf{y}^*))=q_1(\bm{\theta}^*|\bm{\theta}, \mathbf{y})q_2(\mathbf{y}^*|\bm{\theta}^*, \bm{\theta}, \mathbf{y}).
\label{jointproposalAVM}
\end{equation}
Here, the first component $q_1$ can be chosen as a simple random walk proposal that does not depend on $\mathbf{y}$, allowing us to set $q_1(\bm{\theta}^*|\bm{\theta},\mathbf{y})=q(\bm{\theta}^*|\bm{\theta})$. The second component $q_2$ is specified independently of the current state $(\bm{\theta}, \mathbf{y})$ as $q_2(\mathbf{y}^*|\bm{\theta}^*, \bm{\theta}, \mathbf{y})=h(\mathbf{y}^* | \bm{\theta}^*)/Z(\bm{\theta}^*)$. By substituting the joint posterior in \eqref{jointpostAVM} and the factorized proposal in \eqref{jointproposalAVM} into the MH acceptance ratio, we obtain:
\begin{equation}
\alpha_{\mathrm{M{\o}ller}}((\bm{\theta}, \mathbf{y}) \rightarrow (\bm{\theta}^*, \mathbf{y}^*)) = \min \left\{ 1, \frac{f(\mathbf{y}^* | \bm{\theta}^*, \mathbf{x}) p(\bm{\theta}^*) h(\mathbf{x} | \bm{\theta}^*) \cancel{Z(\bm{\theta})} h(\mathbf{y} | \bm{\theta}) \cancel{Z(\bm{\theta}^*)} q(\bm{\theta} | \bm{\theta}^*)}
{f(\mathbf{y} | \bm{\theta}, \mathbf{x}) p(\bm{\theta}) h(\mathbf{x} | \bm{\theta}) \cancel{Z(\bm{\theta}^*)} h(\mathbf{y} | \bm{\theta}^*) \cancel{Z(\bm{\theta})} q(\bm{\theta}^* | \bm{\theta})}
\right\}.
\label{accprobAVM}
\end{equation}
\eqref{accprobAVM} does not contain the intractable terms. We can obtain the original target posterior $\pi(\bm{\theta}|\mathbf{x})$ by taking the marginal samples of $\bm{\theta}$. Note that the mixing of the algorithm depends on the choice of $f(\mathbf{y}|\bm{\theta}, \mathbf{x})$ \citep{moller2006}. Ideally, if we set $f(\mathbf{y}|\bm{\theta}, \mathbf{x})=h(\mathbf{y}|\bm{\theta})/Z(\bm{\theta})$, then \eqref{accprobAVM} becomes equivalent to the acceptance probability of the standard MH algorithm with the stationary distribution $\pi(\bm{\theta}|\mathbf{x})$. However, since $Z(\bm{\theta})$ is intractable, this choice is not feasible in practice. Instead, \cite{moller2006} suggest using $f(\mathbf{y} | \bm{\theta}, \mathbf{x})=h(\mathbf{y}|\widehat{\bm{\theta}})/Z(\widehat{\bm{\theta}})$, where $\widehat{\bm{\theta}}$ can be obtained using the maximum pseudolikelihood estimate (MPLE) \citep{besag1974spatial} or the Markov chain Monte Carlo maximum likelihood estimate (MCMC-MLE) \citep{geyer1992constrained}.

Building on the work of \cite{moller2006}, \cite{Murray2006} proposed an exchange algorithm that eliminates the need to estimate model parameters before running the MCMC algorithm. Let $\mathbf{y}$ be an auxiliary variable that follows $h(\mathbf{y} |\bm{\theta}^*)/Z(\bm{\theta}^*)$. We denote the conditional distribution of $\bm{\theta}^*$ given $\bm{\theta}$ as $q(\bm{\theta}^*|\bm \theta)$, where $\bm{\theta}$ and $\bm{\theta}^*$ are parameters associated with the data $\mathbf{x}$ and $\mathbf{y}$, respectively. For $q(\bm{\theta}^*|\bm \theta)$, we can use a random walk distribution centered at $\bm \theta$. Then, the augmented posterior can be written as
\begin{equation}
    \pi(\bm \theta,\bm{\theta}^*,\mathbf {y}|\mathbf x)\propto p(\bm \theta) \frac{h(\mathbf x|\bm \theta)}{Z(\bm \theta)}q(\bm{\theta}^*|\bm \theta)\frac{h(\mathbf{y} |\bm{\theta}^*)}{Z(\bm{\theta}^*)}.
    \label{jointpost}
\end{equation}
For the augmented posterior, $( \bm{\theta}^*,\mathbf{y} )$ is updated using a block Gibbs sampler. Specifically, we first draw $\bm{\theta}^* \sim q(\bm{\theta}^* | \bm{\theta})$, and then sample $\mathbf{y} \sim h(\mathbf{y} |\bm{\theta}^*)/Z(\bm{\theta}^*)$. Lastly, $\bm{\theta}$ is updated by exchanging parameter settings. The swapping proposal suggests that $\mathbf{x}$ is associated with $\bm{\theta}^*$ and $\mathbf{y}$ is associated with $\bm{\theta}$. Since the swapping proposal is symmetric, the MH acceptance probability for the exchange move from $\bm{\theta}$ to the proposed $\bm{\theta}^\ast$ is 
\begin{align}
\alpha_{\mathrm{AVM}}(\bm \theta \rightarrow \bm{\theta}^*; \mathbf{y})
&= \min \left\{ 1, 
\frac{p(\bm{\theta}^*) h(\mathbf{x} | \bm{\theta}^*) \cancel{Z(\bm{\theta})} h(\mathbf{y} | \bm{\theta}) \cancel{Z(\bm{\theta}^*)} q(\bm{\theta} | \bm{\theta}^*)}
{p(\bm{\theta}) h(\mathbf{x} | \bm{\theta}) \cancel{Z(\bm{\theta}^*)} h(\mathbf{y} | \bm{\theta}^*) \cancel{Z(\bm{\theta})} q(\bm{\theta}^* | \bm{\theta})}
\right\}.
\label{auxMCMCacc}
\end{align} 
Similar to \eqref{accprobAVM}, all the intractable terms are canceled out in \eqref{auxMCMCacc}. The marginal samples of $\bm{\theta}$ follow the target posterior distribution. The AVM algorithm targets the augmented posterior $\pi(\bm{\theta}, \bm{\theta}^*, \mathbf{y}|\mathbf{x})$, while our primary interest lies in the marginal posterior $\pi(\bm{\theta}|\mathbf{x})$, obtained by retaining only the samples of $\bm{\theta}$. Following the expression in the proof of Corollary 2.3 in the appendix of \cite{alquier2014noisy}, the marginal transition kernel of the AVM algorithm is
\begin{equation}
K_{\mathrm{AVM}}(\bm{\theta}, \mathbf{A}) = \int_{\mathbf{A}} q(\bm \theta^* | \bm\theta)\,\int_{\mathcal{X}}{\alpha_{\mathrm{AVM}}(\bm{\theta} \rightarrow \bm{\theta}^* ; \mathbf{y}) \frac{h(\mathbf{y}|\bm{\theta}^*)} {Z(\bm{\theta}^*)} d \mathbf{y}} d\bm{\theta}^* + (1-r_{\mathrm{AVM}}(\bm{\theta})) \mathbf{1}_{\mathbf{A}}(\bm{\theta}),
\label{kernel_AVM}
\end{equation}
where $\mathbf{A} \subseteq \mathcal{B}(\bm{\Theta})$ is a measurable set in the Borel $\sigma$-algebra on the parameter space, and
\[
r_{\mathrm{AVM}}(\bm{\theta})=\int_{\bm{\Theta}}{q(\bm{\theta}^*|\bm{\theta})\int_{\mathcal{X}}{\alpha_{\mathrm{AVM}}(\bm{\theta} \rightarrow \bm{\theta}^* ; \mathbf{y}) \frac{h(\mathbf{y}|\bm{\theta}^*)} {Z(\bm{\theta}^*)} d \mathbf{y}} d\bm{\theta}^*}.
\]
This marginal Markov transition kernel is equivalent to the projection of the Markov transition kernel on the augmented space onto the marginal space. We provide the details in \ref{deriv_marginal}.

Note that both \cite{moller2006} and \cite{Murray2006} obtain the target posterior $\pi(\bm{\theta}|\mathbf{x})$ as the marginal of the joint posterior—$\pi(\bm{\theta}, \mathbf{y}|\mathbf{x})$ in the case of \cite{moller2006}, and $\pi(\bm{\theta}, \bm{\theta}^*, \mathbf{y}|\mathbf{x})$ in the case of \cite{Murray2006}. However, this requires that $\mathbf{y}$ be sampled exactly from the probability model. Although they employed a perfect sampler \citep{propp1996exact} to generate exact samples of $\mathbf{y}$, such samplers are only available in limited cases. To address this, \cite{liang2010double} developed the double Metropolis-Hastings (DMH) algorithm. \cite{liang2010double} generates $\mathbf{y}$ approximately from the probability model by using a standard MCMC sampler, the so-called inner sampler. The remaining steps proceed in the same manner as in the exchange algorithm. Although the stationary distribution of the DMH algorithm is not exactly $\pi(\bm{\theta}, \bm{\theta}^*, \mathbf{y} | \mathbf{x})$ because $\mathbf{y}$ is approximately sampled from the probability model, the DMH samplers have been widely used in many applications \citep{attrep, park2022bayesian, kang2024fast}. The DMH samplers can provide a reliable approximation to the posterior with the appropriate length of the inner sampler \citep{park2018bayesian}. However, the DMH algorithm becomes computationally expensive for large data because auxiliary variable sampling requires a longer chain with increasing data space dimension.

\subsection{Delayed Acceptance MCMC}\label{subsec:DA}

The DA-MCMC method \citep{christen2005markov} can accelerate the MH algorithm, particularly when the likelihood evaluation is computationally expensive. Given a current $\bm{\theta}$, a candidate  $\bm{\theta}^*$ is proposed from $q(\bm{\theta}^* | \bm{\theta})$. Then, the acceptance probability of the first stage of the algorithm is  
\begin{equation}
\alpha_{\mathrm{DA1}}(\bm \theta \rightarrow \bm \theta^*) = \min \left\{ 1, 
\frac{\widehat{\pi}(\bm{\theta}^*|\mathbf x) q(\bm{\theta}| \bm{\theta}^*)}{\widehat{\pi}(\bm{\theta}|\mathbf x) q(\bm{\theta}^* | \bm{\theta})}
\right\},
\label{DAMCMCacc1}
\end{equation}
where $\widehat{\pi}(\bm{\theta}^*|\mathbf x)$ is a computationally cheap approximation to $\pi(\bm{\theta}^*|\mathbf x)$. A variety of DA-MCMC algorithms have been developed by constructing $\widehat{\pi}(\bm{\theta}|\mathbf x)$ through a divide-and-conquer strategy \citep{banterle2015accelerating}, adaptive k-nearest neighbors \citep{sherlock2017adaptive} and subsampling strategy \citep{banterle2015accelerating,quiroz2018speeding}. If $\bm{\theta}^*$ is accepted from \eqref{DAMCMCacc1}, the second stage acceptance probability is computed as follows: 
\begin{equation}
\alpha_{\mathrm{DA2}}(\bm \theta \rightarrow \bm \theta^*) = \min \left\{ 1, 
\frac{\pi(\bm{\theta}^*|\mathbf x)\widehat{\pi}(\bm{\theta}|\mathbf x)}{\pi(\bm{\theta}|\mathbf x)\widehat{\pi}(\bm{\theta}^*|\mathbf x)}
\right\}.
\label{DAMCMCacc2}
\end{equation}
The overall acceptance probability $\alpha_{\mathrm{DA1}}\alpha_{\mathrm{DA2}}$ satisfies the detailed balance condition with respect to $\pi(\bm{\theta}|\mathbf x)$. Since the procedure early rejects $\bm{\theta}^*$ without computing expensive $\pi(\bm{\theta}|\mathbf x)$, the DA-MCMC algorithm can explore the parameter space more effectively. The DA algorithm targets the posterior distribution $\pi(\bm{\theta}|\mathbf{x})$, with the Markov transition kernel
\begin{equation}
    K_{\mathrm{DA}}(\bm{\theta}, \mathbf{A}) = \int_{\mathbf{A}}q(\bm{\theta}^*|\bm{\theta}) \alpha_{\mathrm{DA1}}(\bm{\theta}\rightarrow\bm{\theta}^*) \alpha_{\mathrm{DA2}}(\bm{\theta}\rightarrow\bm{\theta}^*)d\bm{\theta}^* + (1-r_{\mathrm{DA}}(\bm{\theta}))\mathbf{1}_{\mathbf{A}}(\bm{\theta}),
    \label{kernel_DA}
\end{equation}
    where $\mathbf{A} \subseteq \mathcal{B}(\bm{\Theta})$ is a measurable set in the Borel $\sigma$-algebra on the parameter space, and $$r_{\mathrm{DA}}(\bm{\theta})=\int_\Theta q(\bm{\theta}^*|\bm{\theta})\alpha_{\mathrm{DA1}}(\bm{\theta}\rightarrow\bm{\theta}^*)\alpha_{\mathrm{DA2}}(\bm{\theta}\rightarrow\bm{\theta}^*)d\bm{\theta}^*.$$ Here, $\alpha_{\mathrm{DA1}}(\bm{\theta} \rightarrow \bm{\theta}^*)$ and $\alpha_{\mathrm{DA2}}(\bm{\theta} \rightarrow \bm{\theta}^*)$ denote the acceptance probabilities defined in \eqref{DAMCMCacc1} and \eqref{DAMCMCacc2}, respectively. Note that $1-r_{\mathrm{DA}}(\bm{\theta})$ is the overall rejection probability, which can be decomposed into two parts: rejection at the first stage, $1-\int_\Theta q(\bm{\theta}^*|\bm{\theta})\alpha_{\mathrm{DA1}}(\bm{\theta}\rightarrow\bm{\theta}^*)d\bm{\theta}^*$ and rejection at the second stage after passing the first, $\int_\Theta q(\bm{\theta}^*|\bm{\theta})\alpha_{\mathrm{DA1}}(\bm{\theta}\rightarrow\bm{\theta}^*)[1-\alpha_{\mathrm{DA2}}(\bm{\theta}\rightarrow\bm{\theta}^*)]d\bm{\theta}^*$. The sum of these two terms equals $1-r_{\mathrm{DA}}(\bm{\theta})$. However, direct application of the DA algorithm is infeasible since $\alpha_{\mathrm{DA2}}(\bm{\theta}\rightarrow\bm{\theta}^*)$ involves $Z(\bm{\theta})$, which is intractable.  

Due to its flexibility, the DA procedures have also been studied in the approximate Bayesian computation (ABC) literature \citep{Beaumont2002} when the likelihood evaluation is intractable. Given $\bm{\theta}^* \sim q(\bm{\theta}^* | \bm{\theta})$, the ABC methods simulate synthetic data from the probability model. If the discrepancy between the synthetic data and the observed data is small, $\bm{\theta}^*$ is accepted and is used to approximate $\pi(\bm{\theta}|\mathbf x)$. \cite{everitt2021delayed} incorporated DA-MCMC into the ABC sequential Monte Carlo (ABC-SMC) to reduce the expensive synthetic data simulation from the probability model. Recently, \cite{cao2024using} proposed an early rejection algorithm based on the Gaussian process discrepancy model. Motivated by these recent approaches, we propose a DA-AVM in the following section.

\section{Delayed Acceptance Auxiliary Variable MCMC}
\label{sec:3}

In this section, we describe a general framework for DA-AVM for intractable spatial models. 

\subsection{DA-AVM Algorithms}

The DA-AVM algorithm is computationally efficient compared to the standard AVM by reducing the number of auxiliary variable simulations through the initial screening step. The general form of the first stage kernel can be defined as \eqref{DAMCMCacc1}. In Section~\ref{firstkernel}, we provide details for constructing the first stage kernel. Specifically, we construct a computationally efficient surrogate $\widehat{\pi}(\bm{\theta}|\mathbf{x})$ using subsampling, function emulation, and frequentist estimator–based approximation. Once $\bm{\theta}^*$ is accepted in the first stage, we generate $\mathbf{y} \sim h(\mathbf{y}|\bm{\theta}^*)/Z(\bm{\theta}^*)$ in the second stage. 
From this procedure, we can avoid the simulation of the auxiliary variable if $\bm{\theta}^*$ belongs to the implausible region of the parameter space. Note that the efficiency of the DA-AVM is affected by the surrogate model construction. If $\widehat{\pi}(\bm{\theta}|\mathbf x)$ cannot approximate the true $\pi(\bm{\theta}|\mathbf x)$ well, the algorithm can reject a good candidate in the first stage, resulting in the slow mixing of the chain. Furthermore, if $\widehat{\pi}(\bm{\theta}|\mathbf x)$ is non-informative (i.e., too flat), most proposals are likely to be accepted in the first stage; therefore, the computational savings are marginal. In Section~\ref{sec:simulation}, we compare the efficiency of different surrogate models and discuss practical implementation issues.

\subsection{First Stage Kernel Construction}\label{firstkernel}
\subsubsection{Subsampling Strategy}\label{DAsub} 

There have been several proposals to construct the first stage kernel using subsampling strategies. For instance, \cite{banterle2015accelerating} split the Metropolis-Hastings acceptance step into multiple components and evaluated them sequentially to allow early rejection. \cite{quiroz2018speeding} approximated the likelihood based on a random subsample in the first stage of the DA-MCMC algorithm and reduced the variance of the approximated likelihood using control variates. In a similar fashion, we propose $\text{DA-AVM}_{\mathrm{S}}$ based on a subsampling strategy.

When subsampling spatial data, it is important to preserve the local spatial dependence structure. To achieve this, we sample $\mathbf{x}_{\mathrm{sub}} \in \mathcal{X}_{\mathrm{sub}} \subset \mathcal{X}$, where $\mathcal{X}_{\mathrm{sub}}$ denotes a subregion of the data space; for each iteration in the first stage kernel, a subset of the data is selected. Given a generated $\bm{\theta}^*$, an auxiliary variable $\mathbf{y}_{\mathrm{sub}}$ is sampled from $h(\mathbf{y}_{\mathrm{sub}}|\bm\theta^*)/Z(\bm{\theta}^*)$. The acceptance probability of the first stage kernel with the subset of the dataset is 
\begin{equation}
    \alpha_{\mathrm{S1}}(\bm{\theta} \rightarrow \bm{\theta}^* ; \mathbf{y}_{\mathrm{sub}})=\min\left\{1, \frac{p(\bm\theta^*)h(\mathbf x_{\mathrm{sub}}| \bm \theta^*)h(\mathbf y_{\mathrm{sub}}| \bm \theta) q(\bm\theta | \bm\theta^*)}{p(\bm\theta)h(\mathbf x_{\mathrm{sub}}|\bm \theta)h(\mathbf y_{\mathrm{sub}}| \bm \theta^*) q(\bm\theta^* | \bm\theta)} \right\}.
    \label{subsam_alpha1}
\end{equation}
Since $\mathbf{y}_{\mathrm{sub}}$ has the same dimension as $\mathbf{x}_{\mathrm{sub}}$, which is much smaller than that of $\mathbf{x}$, the auxiliary variable simulation becomes much faster. The length of the inner sampler for generating $\mathbf{y}_{\mathrm{sub}}$ can be substantially shorter than that for $\mathbf{y}$. Once $\bm{\theta}^*$ is accepted, we generate an auxiliary variable $\mathbf{y} \sim h(\mathbf{y}|\bm\theta^*)/Z(\bm{\theta}^*)$ and the acceptance probability of the second stage kernel becomes 
\begin{equation}
     \alpha_{\mathrm{S2}}(\bm{\theta} \rightarrow \bm{\theta}^* ; \mathbf{y})=\min \left\{ 1, \frac{ h(\mathbf{x}_{\mathrm{sub}} | \bm{\theta}) h(\mathbf y_{\mathrm{sub}} | \bm \theta^*)h(\mathbf x|\bm \theta^*)h(\mathbf y| \bm \theta)}{h(\mathbf{x}_{\mathrm{sub}} | \bm{\theta^*}) h(\mathbf{y}_{\mathrm{sub}} | \bm \theta)h(\mathbf x| \bm \theta)h(\mathbf y | \bm \theta^*)}\right\}.    
\end{equation}
An advantage of the proposed methodology is that it requires fewer components to be tuned in surrogate model construction compared to other approaches. Once an inner sampler for generating auxiliary variables from the probability model is available, only minor adjustments are needed to generate $\mathbf{y}_{\mathrm{sub}} \in \mathcal{X}_{\mathrm{sub}}$. We provide algorithm details for $\text{DA-AVM}_{\mathrm{S}}$ in \ref{algo} (Algorithm~\ref{daavsp}).

The efficiency of the algorithm depends on the size of $\mathbf{x}_{\mathrm{sub}}$. If $\mathbf{x}_{\mathrm{sub}}$ is too small, the approximate posterior in the first stage becomes non-informative, leading to most proposals being accepted in the first stage. Consequently, auxiliary variable simulations must be performed twice (in both the first and second stages), and the computational savings may become negligible. In Section~\ref{sec:simulation}, we observe that using a subset approximately one-fourth the size of the full data is efficient, particularly in cases such as point process models where the computational complexity of the inner sampler is quadratic.

\subsubsection{Function Emulation Approach}
\label{fesection}

Gaussian process emulations have been widely used to accelerate inference for models with intractable likelihood functions \citep{drovandi2018accelerating, park2020function, vu2023warped}. In this work, we utilize a function emulation approach \citep{park2020function} to construct the first stage kernel of DA-AVM. 

Let $\bm{\theta}^{(1)},\cdots,\bm{\theta}^{(d)}$ denote the $d$ particles that cover $\bm\Theta \subset \mathbb R^{p}$. As $p$ increases, the particles must be carefully designed to cover the important region of $\bm{\Theta}$. Otherwise, a substantially larger number of particles $d$ would be required, which can affect computational efficiency. Following \cite{park2020function}, we construct the particles by using the ABC algorithm or the short run of the AVM algorithm. The logarithm of the importance sampling estimate at $\bm\theta^{(i)}$ is 
\begin{equation}
\log \widehat{Z}_{\mathrm{IS}}(\bm\theta^{(i)}) = \log \left( \frac{1}{N} \sum_{l=1}^{N} \frac{h(\mathbf{x}_l| \bm\theta^{(i)})}{h(\mathbf{x}_l | \widetilde{\bm\theta}) }\right),
\label{importanceest}
\end{equation}
where $\lbrace \mathbf{x}_l \rbrace_{l=1}^{N}$ are samples generated from a Markov chain whose stationary distribution is $h(\cdot | \widetilde{\bm\theta}) / Z(\widetilde{\bm\theta})$. Here, $\widetilde{\bm\theta}$ can be an approximation to the MLE or the maximum pseudo-likelihood estimator (MPLE). Let $\log \mathbf{\widehat{Z}_{\mathrm{IS}}}=(\log \widehat{Z}_{\mathrm{IS}}(\bm\theta^{(1)}),\cdots,\log \widehat{Z}_{\mathrm{IS}}(\bm\theta^{(d)}))' \in \mathbb R^{d}$ be a vector of the log importance sampling estimates evaluated at each particle. Then we can define a Gaussian process regression model as
\begin{equation}
\log \mathbf{\widehat{Z}_{\mathrm{IS}}} = \bm{\Psi} \bm\beta + \mathbf{W},
\label{gpest}
\end{equation}
where $\bm{\Psi} \in \mathbb R^{d \times p}$ is the design matrix whose rows consist of the particles, and $\bm{\beta} \in \mathbb{R}^{p}$ denotes the regression coefficients. In \eqref{gpest}, $\mathbf{W} \in \mathbb R^{d}$ follows a zero-mean second order stationary Gaussian process with a Mat\'{e}rn class \citep{stein2012interpolation} covariance function defined as
\begin{equation}
\mathbf{K}(\bm\theta^{(i)}, \bm\theta^{(j)}; \sigma^2, \phi, \tau^2) = \sigma^2\Big(  1 + \frac{\sqrt{3}\|\bm{\theta}^{(i)}-\bm{\theta}^{(j)}\|}{\phi} \Big) \exp{\Big( -\frac{\sqrt{3}\|\bm{\theta}^{(i)}-\bm{\theta}^{(j)}\|}{\phi}   \Big)} + \tau^2 1_{\lbrace i = j \rbrace}.
\label{matern}
\end{equation}	
Here, $\sigma^2$, $\phi$, and $\tau$ denote the partial sill, the range parameter, and the nugget, respectively. To interpolate $\log \widehat{Z}_{\mathrm{IS}}(\bm{\theta}^{\ast})$ for an arbitrary $\bm{\theta}^{\ast} \in \mathbb R^{p}$, we define the joint distribution as
\begin{equation}
\begin{bmatrix}
\log \widehat{\mathbf{Z}}_{\mathrm{IS}} \\
\log \widehat{Z}_{\mathrm{IS}}(\bm{\theta}^{\ast})
\end{bmatrix}
=
N \Bigg(
\begin{bmatrix}
\bm{\Psi} \bm{\beta} \\
\bm{\theta}^{\ast '} \bm{\beta}
\end{bmatrix},
\begin{bmatrix}
\mathbf{C} & \mathbf{c} \\
\mathbf{c}' & \sigma^2 + \tau^2
\end{bmatrix}
\Bigg),
\end{equation}
with $\mathbf{C}=\mathbf{K}(\bm{\Psi}, \bm{\Psi}; \sigma^2, \phi,\tau) \in \mathbb{R}^{d \times d}$ and $\mathbf{c}=\mathbf{K}(\bm{\Psi}, \bm{\theta}^{\ast}; \sigma^2, \phi,\tau) \in \mathbb{R}^{d \times 1}$. Then the conditional distribution of $\log \widehat{Z}_{\mathrm{IS}}(\bm{\theta}^{\ast})$ given  $\log \mathbf{\widehat{Z}_{\mathrm{IS}}}$ is
\begin{equation}
\log \widehat{Z}_{\mathrm{IS}}(\bm{\theta}^{\ast})|\log \mathbf{\widehat{Z}_{\mathrm{IS}}} \sim N( \bm{\theta}^{\ast '}\bm{\beta}  + \mathbf{c}'\mathbf{C}^{-1}(\mathbf{\log\widehat{Z}_{\mathrm{IS}}}-\bm{\Psi}\bm{\beta} ), \sigma^{2}+\tau^2-\mathbf{c}'\mathbf{C}^{-1}\mathbf{c}).
\label{condexp}
\end{equation}
We obtain the empirical best linear unbiased predictor (EBLUP) for $\log \widehat{Z}_{\mathrm{IS}}(\bm{\theta}^{\ast})$ as
\begin{equation}
\log \widehat{Z}_{\mathrm{GP}}(\bm{\theta}^{\ast}) =\bm{\theta}^{\ast '}\widehat{\bm{\beta}} + \widehat{\mathbf{c}}'\widehat{\mathbf{C}}^{-1}(\mathbf{\log \widehat{Z}_{\mathrm{IS}}} -\bm{\Psi}\widehat{\bm{\beta}})  
\label{BLUP1}
\end{equation}
by plugging in the estimates of the covariance parameters $(\sigma^2, \phi, \tau)$ and the GLS estimate $\widehat{\bm{\beta}}$. In our study, we use the {\tt DiceKriging} package \citep{roustant2012dicekriging} to fit a Gaussian process regression model. 

We propose $\text{DA-AVM}_{\mathrm{GP}}$ by constructing the surrogate as $\widehat{\pi}_{\mathrm{GP}}(\bm{\theta}|\mathbf x) \propto p(\bm{\theta})h(\mathbf{x}|\bm{\theta})/\widehat{Z}_{\mathrm{GP}}(\bm\theta)$. Then, the acceptance probability of the first stage kernel is 
\begin{equation}
    \alpha_{\mathrm{GP1}}(\bm \theta \rightarrow \bm{\theta}^*)=\min\left\{1, \frac{\widehat{\pi}_{\mathrm{GP}}(\bm{\theta}^*|\mathbf x)q(\bm{\theta}| \bm{\theta}^*)}{\widehat{\pi}_{\mathrm{GP}}(\bm{\theta}|\mathbf x) q(\bm{\theta}^* | \bm{\theta})} \right\}.
    \label{gp_alpha1}
\end{equation}
Once fitted, the Gaussian process emulation can evaluate \eqref{gp_alpha1} very quickly. Note that the Gaussian process emulator is precomputed prior to running the MCMC algorithm. To reduce the computational cost, parallel computation is employed to construct the importance sampling estimate in \eqref{importanceest}. Subsequently, fitting the EBLUP takes only a few seconds. Once $\bm{\theta}^*$ is accepted, we generate $\mathbf{y} \sim h(\mathbf{y}|\bm{\theta}^*)/Z(\bm{\theta}^*)$. Then the MH acceptance probability of the second stage becomes
\begin{equation}
    \alpha_{\mathrm{GP2}}(\bm{\theta} \rightarrow \bm{\theta}^* ; \mathbf{y})=\min \left\{ 1, \frac{p(\bm{\theta}^*) h(\mathbf{x} | \bm{\theta}^*)h(\mathbf{y} | \bm{\theta})\widehat{\pi}_{\mathrm{GP}}(\bm{\theta}|\mathbf x)}
{p(\bm{\theta}) h(\mathbf{x} | \bm{\theta})h(\mathbf{y} | \bm{\theta}^*)\widehat{\pi}_{\mathrm{GP}}(\bm{\theta}^*|\mathbf x)}\right\}.
\label{gp_alpha2}
\end{equation}
We provide algorithm details for $\text{DA-AVM}_{\mathrm{GP}}$ in \ref{algo} (Algorithm~\ref{daavgp}).

\subsubsection{Frequentist Estimator-Based Approximation}

Frequentist computational methods have been developed for several classes of spatial models, including lattice models \citep{potts1952some} and ERGMs \citep{robins2007introduction}. We construct the first stage kernel of DA-AVM based on such frequentist estimators.

The pseudo-likelihood approach \citep{besag1974spatial} approximates the likelihood function using a simplified form by ignoring certain levels of spatial dependencies. Specifically, the logarithm of the pseudo-likelihood function is defined as
\begin{equation}
\log \text{PL}(\bm{\theta};\mathbf{x}) = \sum_{i=1}^{n} \log p(x_i|\mathbf{x}_{-i},\bm{\theta}),
\label{pseudolik}
\end{equation}
where $p(x_i|\mathbf{x}_{-i},\bm{\theta})$ is a full conditional distribution. Since \eqref{pseudolik} does not involve the intractable normalizing function $Z(\bm{\theta})$, the MPLE can be easily obtained. The MPLE can be a practical option when the spatial dependency among $\mathbf{x}$ is relatively weak. Alternatively, the Monte Carlo maximum likelihood (MCML) method \citep{geyer1992constrained} has been applied to a wide variety of applications. Based on the importance sampling estimate \eqref{importanceest}, the Monte Carlo maximum likelihood estimator (MCMLE) can be obtained by maximizing the following approximated likelihood function:
\begin{equation}
\log \widehat{L}(\bm{\theta};\mathbf{x}) = \log h(\mathbf{x}|\bm{\theta}) - \log \widehat{Z}_{\mathrm{IS}}(\bm\theta).
\label{MCMLlik}
\end{equation}
If the analytical gradient of $h(\mathbf{x}|\bm{\theta})$ is available, as in ERGMs or spatial lattice models, the MCMLE can be obtained efficiently. In general, the MCMLE provides more accurate inference results than the MPLE because \eqref{MCMLlik} does not ignore spatial dependencies.

We propose $\text{DA-AVM}_{\mathrm{F}}$ by constructing the surrogate based on the asymptotic distribution obtained from frequentist estimators (i.e., the MPLE or MCMLE). Specifically, $\widehat{\pi}_\mathrm{F}(\bm{\theta}|\mathbf{x})$ is obtained as the density of a normal distribution with the mean given by the MPLE or MCMLE and the covariance given by the corresponding observed Fisher information. Then, the acceptance probability of the first stage kernel is 
\begin{equation}
    \alpha_{\mathrm{F1}}(\bm \theta \rightarrow \bm \theta^*)=\min\left\{1, \frac{\widehat{\pi}_\mathrm{F}(\bm{\theta}^*|\mathbf x) q(\bm{\theta}| \bm{\theta}^*)}{\widehat{\pi}_\mathrm{F}(\bm{\theta}|\mathbf x) q(\bm{\theta}^* | \bm{\theta})} \right\}.
    \label{freq_alpha1}
\end{equation} 
Once $\bm{\theta}^*$ is accepted, we generate $\mathbf{y} \sim h(\mathbf{y}|\bm{\theta}^*)/Z(\bm{\theta}^*)$. Similarly, the MH acceptance probability of the second stage kernel becomes
\begin{equation}
     \alpha_{\mathrm{F2}}(\bm{\theta} \rightarrow \bm{\theta}^* ; \mathbf{y})=\min \left\{ 1, \frac{p(\bm{\theta}^*) h(\mathbf{x} | \bm{\theta}^*)h(\mathbf{y} | \bm{\theta})\widehat{\pi}_\mathrm{F}(\bm{\theta}|\mathbf x)}
{p(\bm{\theta}) h(\mathbf{x} | \bm{\theta})h(\mathbf{y} | \bm{\theta}^*)\widehat{\pi}_\mathrm{F}(\bm{\theta}^*|\mathbf x)}\right\}.
\label{freq_alpha2}
\end{equation} 
We provide algorithm details for $\text{DA-AVM}_{\mathrm{F}}$ in \ref{algo} (Algorithm~\ref{daavn}).

\subsection{Properties of DA-AVM}  

In this section, we show the theoretical properties of $\text{DA-AVM}_{\mathrm{F}}$. Similar results can also be derived in the context of the function emulation approach or the subsampling strategy. We first show that $\text{DA-AVM}_{\mathrm{F}}$ satisfies the detailed balance condition, ensuring that the stationary distribution induced by the $\text{DA-AVM}_{\mathrm{F}}$ algorithm is identical to that of the standard AVM. Let $\alpha_{\mathrm{F1}}$ and $\alpha_{\mathrm{F2}}$ denote the acceptance probabilities associated with the first and second stages of the kernel, as defined in \eqref{freq_alpha1} and \eqref{freq_alpha2}, respectively. Then, the detailed balance condition with respect to the marginal posterior distribution holds as follows:
\begin{align}
\pi(\bm{\theta}| \mathbf{x})&\overbrace{q(\bm{\theta}^*|\bm{\theta})\alpha_{\mathrm{F1}}(\bm{\theta} \rightarrow \bm{\theta}^*)}^{\text{first stage}}\overbrace{\mathbf{E}_{\mathbf{y}}\left[\alpha_{\mathrm{F2}}(\bm{\theta} \rightarrow \bm{\theta}^* ; \mathbf{y})\right]}^{\text{second stage}}\nonumber\\
&=\pi(\bm{\theta}| \mathbf{x})q(\bm{\theta}^*|\bm{\theta})\alpha_{\mathrm{F1}}(\bm{\theta} \rightarrow \bm{\theta}^*)\int{\alpha_{\mathrm{F2}}(\bm{\theta} \rightarrow \bm{\theta}^* ; \mathbf{y})}\frac{h(\mathbf{y}|\bm{\theta}^*)}{Z(\bm{\theta}^*)} d\mathbf{y}\nonumber\\
&=\pi(\bm{\theta}|\mathbf{x}) q(\bm{\theta}^* | \bm{\theta}) \int \min\left\{1,\frac{ \widehat{\pi}_\mathrm{F}(\bm{\theta}^* | \mathbf{x}) q(\bm{\theta} | \bm{\theta}^*)}{ \widehat{\pi}_\mathrm{F}(\bm{\theta} | \mathbf{x}) q(\bm{\theta}^* | \bm{\theta})}\right\} \nonumber\\
&\quad \times \min\left\{1,\frac{h(\mathbf{y}|\bm{\theta})p(\bm \theta^*)h(\mathbf{x} | \bm{\theta}^*) \widehat{\pi}_\mathrm{F}(\bm{\theta} | \mathbf{x})}{h(\mathbf{x}|\bm{\theta})p(\bm \theta)  h(\mathbf{y} | \bm{\theta}^*) \widehat{\pi}_\mathrm{F}(\bm{\theta}^* | \mathbf{x})} \right\}  \frac{h(\mathbf{y}|\bm{\theta}^*)}{Z(\bm{\theta}^*)} d\mathbf{y} \nonumber\\
&=\pi(\bm{\theta}^*|\mathbf{x}) q(\bm{\theta} | \bm{\theta}^*) \int\min\left\{1,\frac{ \widehat{\pi}_\mathrm{F}(\bm{\theta} | \mathbf{x}) q(\bm{\theta}^* | \bm{\theta})}{ \widehat{\pi}_\mathrm{F}(\bm{\theta}^* | \mathbf{x}) q(\bm{\theta} | \bm{\theta}^*)}\right\} \nonumber\\
&\quad \times \min\left\{1,\frac{h(\mathbf{y}^*|\bm{\theta}^*)p(\bm \theta)h(\mathbf{x} | \bm{\theta}) \widehat{\pi}_\mathrm{F}(\bm{\theta}^* | \mathbf{x})}{h(\mathbf{x}|\bm{\theta}^*)p(\bm \theta^*)  h(\mathbf{y}^* | \bm{\theta}) \widehat{\pi}_\mathrm{F}(\bm{\theta} | \mathbf{x})} \right\}  \frac{h(\mathbf{y}^*|\bm{\theta})}{Z(\bm{\theta})} d\mathbf{y}^* \nonumber\\
&=\pi(\bm{\theta}^*| \mathbf{x})q(\bm{\theta}|\bm{\theta}^*)\alpha_{\mathrm{F1}}(\bm{\theta}^* \rightarrow \bm{\theta})\int{\alpha_{\mathrm{F2}}(\bm{\theta} \rightarrow \bm{\theta}^* ; \mathbf{y})\frac{h(\mathbf{y}^*|\bm{\theta})}{Z(\bm{\theta})} d\mathbf{y}^*}\nonumber \\
&=\pi(\bm{\theta}^*| \mathbf{x})q(\bm{\theta}|\bm{\theta}^*)\alpha_{\mathrm{F1}}(\bm{\theta}^* \rightarrow \bm{\theta}) \mathbf{E}_{\mathbf{y}^*} \left[\alpha_{\mathrm{F2}}(\bm{\theta} \rightarrow \bm{\theta}^* ; \mathbf{y})\right]
\label{detailed}
\end{align}
As previously discussed, when the auxiliary variable is approximately generated from the probability model using a standard MCMC sampler (i.e., the inner sampler), the AVM targets an approximation to the joint posterior. Since perfect sampling is not available for many spatial models, we generate the auxiliary variable through an MCMC sampler; therefore, the stationary distribution induced by the DA-AVM algorithm is also the approximation of the joint posterior.  

While detailed balance ensures that the Markov chain has the correct stationary distribution, ergodicity guarantees that the chain will converge to this stationary distribution regardless of the initial state. In Theorem~\ref{thm1}, we show the ergodicity of the $\text{DA-AVM}_{\mathrm{F}}$ algorithm.


\begin{theorem}
\label{thm1}
Let $K_{\mathrm{DA-AVM}}(\cdot, \cdot)$ and $K_{\mathrm{AVM}}(\cdot, \cdot)$ denote the Markov transition kernels for $\text{DA-AVM}_{\mathrm{F}}$ and AVM, respectively. Suppose that $K_{\mathrm{AVM}}(\cdot, \cdot)$ is $\pi$-irreducible, the proposal $q(\cdot|\cdot)$ is reversible, and $q(\bm{\theta}^*|\bm{\theta})>0$ implies $\widehat{\pi}_\mathrm{F}(\bm{\theta}^*|\mathbf{x})>0$. If $K_{\mathrm{AVM}}(\bm{\theta}, \bm{\theta})>0$ implies $K_{\mathrm{DA-AVM}}(\bm{\theta}, \bm{\theta})>0$, then $K_{\mathrm{DA-AVM}}(\cdot, \cdot)$ is ergodic.  
\end{theorem}
\begin{proof}
The Markov transition kernel of the DA-AVM algorithm is obtained by composing the transition kernels of the DA and AVM algorithms. Specifically, in the first stage, the acceptance probability $\alpha_{\mathrm{F1}}(\bm{\theta} \to \bm{\theta}^*)$ inherited from \eqref{kernel_DA} acts as a screening step prior to generating the auxiliary variable $\mathbf{y}$. Conditional on passing the first stage, the second stage follows the marginal transition kernel of the AVM algorithm in \eqref{kernel_AVM}. Let $\mathbf{A} \subseteq \mathcal{B}(\bm{\Theta})$ be a measurable set, where $\mathcal{B}(\bm{\Theta})$ denotes the Borel $\sigma$-algebra on the parameter space. Then, the Markov transition kernel for DA-AVM can be defined as
\begin{align}
K_{\mathrm{DA-AVM}}(\bm{\theta}, \mathbf{A}) &= 
\overbrace{\int_{\mathbf{A}} \underbrace{q(\bm{\theta}^* | \bm{\theta}) \alpha_{\mathrm{F1}}(\bm{\theta} \rightarrow \bm{\theta}^*) }_\text{first stage} \underbrace{\int_{\mathcal{X}} \alpha_{\mathrm{F2}}(\bm{\theta} \rightarrow \bm{\theta}^* ; \mathbf{y})  \frac{h(\mathbf{y}|\bm{\theta}^*)}{Z(\bm{\theta}^*)}  d \mathbf{y}}_\text{second stage}  d\bm{\theta}^*}^\text{Delayed acceptance MH - accept} \nonumber \\
&\quad+ \overbrace{(1- r_{\mathrm{DA-AVM}}(\bm{\theta})) \mathbf{1}_{\mathbf{A}}(\bm{\theta})}^\text{Delayed acceptance MH - reject},
\end{align}
where 
\[
r_{\mathrm{DA-AVM}}(\bm{\theta})=\int_{\bm{\Theta}} q(\bm{\theta}^* | \bm{\theta}) \alpha_{\mathrm{F1}}(\bm{\theta} \rightarrow \bm{\theta}^*) \int_{\mathcal{X}}  \frac{h(\mathbf{y}|\bm{\theta}^*)}{Z(\bm{\theta}^*)} \alpha_{\mathrm{F2}}(\bm{\theta} \rightarrow \bm{\theta}^* ; \mathbf{y}) d \mathbf{y}  d\bm{\theta}^*.
\]
Here, $\alpha_{\mathrm{F1}}(\bm{\theta} \rightarrow \bm{\theta}^*)$ and $\alpha_{\mathrm{F2}}(\bm{\theta} \rightarrow \bm{\theta}^* ; \mathbf{y})$ denote the acceptance probabilities defined in \eqref{freq_alpha1} and \eqref{freq_alpha2}, respectively. To establish the ergodicity of $K_\mathrm{DA-AVM}(\cdot, \cdot)$, it is necessary to verify irreducibility, aperiodicity, and reversibility (see Corollary 2 in \cite{tierney1994markov} and Lemmas 1.1 and 1.2 in \cite{mengersen1996rates}).

\begin{itemize}
    \item Irreducibility: Since $K_{\mathrm{AVM}}(\cdot, \cdot)$ is assumed to be $\pi$-irreducible, there exists $n \in \mathbb{N}$ such that $K_{\mathrm{AVM}}^n(\bm{\theta}, \mathbf{A})>0$ for any $\bm{\theta} \in \bm{\Theta}$. This implies that $$q(\bm{\theta}^*|\bm{\theta}) \int_{\mathcal{X}}\alpha_{\mathrm{AVM}}(\bm{\theta} \rightarrow \bm{\theta}^* ; \mathbf{y}) 
    \frac{h(\mathbf{y}|\bm{\theta}^*)}{Z(\bm{\theta}^*)} d \mathbf{y}  > 0.$$ By construction, if $\alpha_{\mathrm{AVM}}(\bm{\theta} \rightarrow \bm{\theta}^* ; \mathbf{y})>0$, then  $\alpha_{\mathrm{F1}}(\bm{\theta} \rightarrow \bm{\theta}^*)\alpha_{\mathrm{F2}}(\bm{\theta} \rightarrow \bm{\theta}^* ; \mathbf{y})$ are strictly positive. Hence, $$q(\bm{\theta}^*|\bm{\theta}) \alpha_{\mathrm{F1}}(\bm{\theta} \rightarrow \bm{\theta}^*) \int_{\mathcal{X}}\alpha_{\mathrm{F2}}(\bm{\theta} \rightarrow \bm{\theta}^* ; \mathbf{y})\frac{h(\mathbf{y}|\bm{\theta}^*)}{Z(\bm{\theta}^*)} d\mathbf{y} $$ is strictly positive, implying that $K_{\mathrm{DA-AVM}}(\cdot|\cdot)$ is also $\pi$-irreducible. 
    \item Aperiodicity: Aperiodicity of the DA-AVM kernel is guaranteed by Theorem 1 of \cite{christen2005markov}, provided that $K_\mathrm{AVM}(\bm{\theta}, \bm{\theta})>0$ implies $K_\mathrm{DA-AVM}(\bm{\theta}, \bm{\theta})>0$. When $K_\mathrm{AVM}(\bm{\theta}, \bm{\theta})>0$, it follows that the rejection probability satisfies $r_\mathrm{AVM}(\bm \theta)<1$. Since, $\alpha_{\mathrm{AVM}}(\bm{\theta} \rightarrow \bm{\theta}^* ; \mathbf{y})>0$ implies $\alpha_{\mathrm{F1}}(\bm{\theta} \rightarrow \bm{\theta}^*)\alpha_{\mathrm{F2}}(\bm{\theta} \rightarrow \bm{\theta}^* ; \mathbf{y})>0$, the DA-AVM kernel also satisfies $r_{\mathrm{DA-AVM}}(\bm \theta)<1$, and thus $K_\mathrm{DA-AVM}(\bm \theta, \bm \theta)>0$. Therefore, the kernel $K_\mathrm{DA-AVM}(\cdot,\cdot)$ is aperiodic.
    \item Reversibility: $K_{\mathrm{DA-AVM}}(\cdot,\cdot)$ satisfies the detailed balance condition as we showed in \eqref{detailed}.
\end{itemize}  
\end{proof}

\begin{remark}
Theorem~1 explicitly requires that $q(\bm{\theta}^*|\bm{\theta})>0$ implies $\widehat{\pi}_\mathrm{F}(\bm{\theta}^*|\mathbf{x})>0$, to ensure that the delayed-acceptance step does not exclude admissible proposals and thus preserves irreducibility of the chain. 
\end{remark}

In practice, we construct $\widehat{\pi}_{\mathrm{F}}(\bm{\theta}|\mathbf{x})$ as a Gaussian density whose mean is given by a frequentist estimator (e.g., MPLE or MCMLE) and whose covariance is given by the corresponding observed Fisher information. This surrogate is precomputed before running the DA-AVM algorithm and kept fixed throughout the MCMC run. Since a Gaussian density with a positive definite covariance matrix is strictly positive on its support, the above condition is satisfied whenever the covariance estimate is positive definite.

\section{Applications}
\label{sec:simulation}


In this section, we apply the proposed DA methods to four intractable models: three spatial models—the Potts model, the interaction point process model, and the exponential random graph model (ERGM)—and one non-spatial model, the susceptible-infected-recovered (SIR) model. Each dataset used in the examples is illustrated in Figure~\ref{dataplot}. As described, we construct the first stage kernel based on subsampling, Gaussian process emulation, and frequentist estimators, which are denoted by $\text{DA-AVM}_{\mathrm{S}}$, $\text{DA-AVM}_{\mathrm{GP}}$, and $\text{DA-AVM}_{\mathrm{F}}$, respectively. To illustrate the performance of our approaches, we compare the DA-AVM methods with the standard AVM \citep{liang2010double}. Following \cite{cao2024using}, we use 
\[
\text{Eff} = \frac{\text{\# of early rejected parameters}}{\text{\# of rejected parameters}}
\]
to assess the efficiency of the DA-AVM methods. The efficiency value is bounded between 0 and 1, where a value of 1 represents the ideal case in which all rejected parameters are filtered out during the first stage of the algorithm. The code for the applications is implemented in {\tt R} and {\tt C++}, using {\tt Rcpp} and {\tt RcppArmadillo} \citep{eddelbuettel2011rcpp} packages. We use {\tt DiceKriging} package \citep{roustant2012dicekriging} to fit Gaussian process emulator for $\text{DA-AVM}_{\mathrm{GP}}$. All experiments were conducted on a machine equipped with an Apple M3 Pro chip (11-core CPU, 14-core GPU) and 18 GB of RAM, running macOS Sequoia 15.3.2. The source code can be downloaded from the following repository (\url{https://github.com/rlawhdals/DA-AVM}).

\begin{figure}[htbp!]
\centering
\subfloat[
Simulated Potts process with $\theta=0.8$.]{{\includegraphics[width=0.45\linewidth]{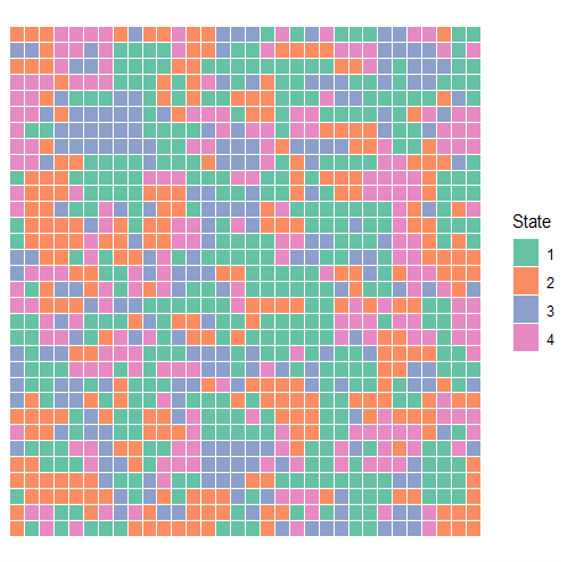} }}%
\hspace{0.05\linewidth}
\subfloat[RSV-A point pattern data collected from the $\mathbf{1A2A}$ experiment \citep{attrep}.]{{\includegraphics[width=0.45\linewidth ]{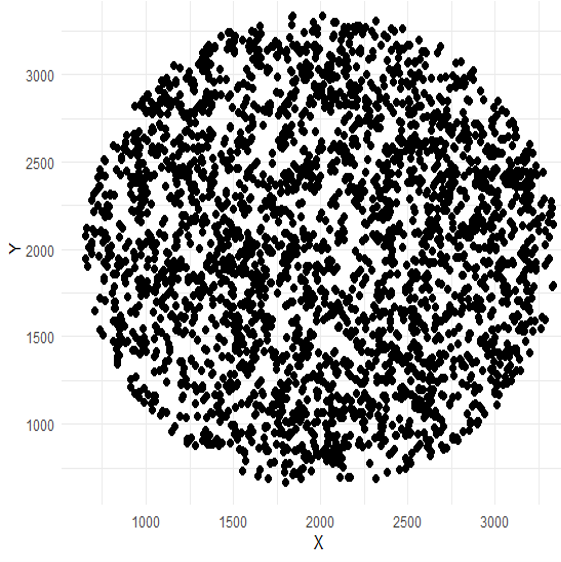} }}%
\hfill 
\subfloat[Faux Mesa high school network dataset \citep{goodreau2007advances, resnick1997protecting}]{{\includegraphics[width=0.45\linewidth ]{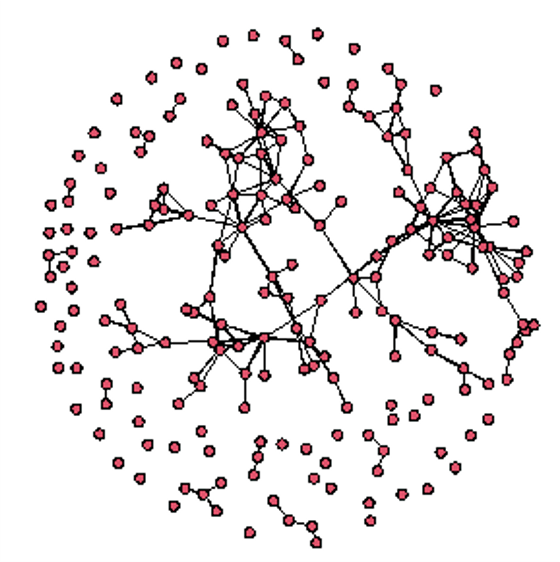}. }}%
\hspace{0.05\linewidth}
\subfloat[Weekly measles incidence data from Baltimore \citep{king2016statistical}.]{{\includegraphics[width=0.45\linewidth ]{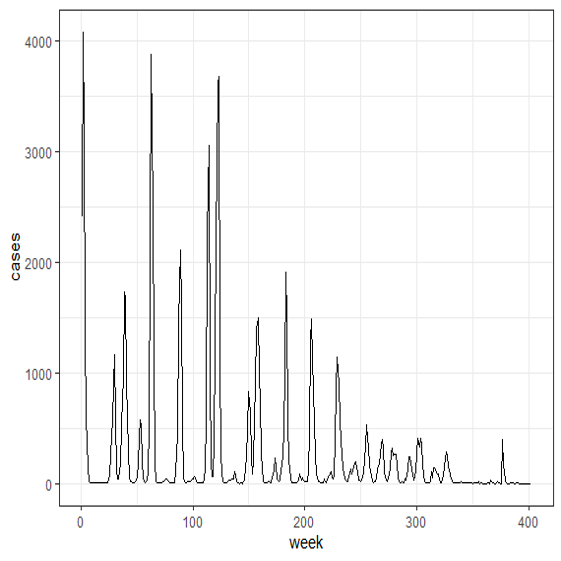} }}%
\caption{Data visualizations for each application.}%
\label{dataplot}%
\end{figure}

\subsection{A Potts Model}

The Potts models \citep{potts1952some} have been widely used to describe spatial interactions with multiple discrete states. For an observed $m \times m$ lattice $\mathbf{x} = \{ x_i \}$ with $x_i \in \{1, \dots, 4\}$, the probability model is 
\begin{equation}
\frac{1}{Z(\theta)}\exp\Big\lbrace \theta \sum_{i \sim j} \delta(x_i, x_j) \Big\rbrace, 
\label{pottseq}
\end{equation}
where $\delta(x_i, x_j)$ is a Kronecker delta function and $i \sim j$ denotes neighboring sites. Here, $\theta \in [0, 2]$ is a parameter that controls the spatial interaction; a larger value of $\theta$ implies a high expected number of neighboring pairs occupying the same state. In \eqref{pottseq}, the computation of $Z(\theta)$ requires summation over all $4^{m\times m}$ possible configurations, which is intractable. We simulate $\mathbf x$ on a $32 \times 32$ lattice with $\theta = 0.8$ using the {\tt potts} package. We use a uniform prior with a range $[0,2]$ for all methods. We run MCMC algorithms for 50,000 iterations until convergence and discard 10,000 samples for burn-in. We generate the auxiliary variable using 10 cycles (i.e., $10\times32\times32$ iterations) of the Gibbs sampler. 

Since the Gaussian process emulator is efficient for low-dimensional parameter problems, we implement $\text{DA-AVM}_{\mathrm{GP40}}$; ${\text{GP40}}$ indicates that the Gaussian process emulator was constructed using 40 particles. We generate particles by using the ABC algorithm described in \ref{algo} (Algorithm~\ref{ABCalg}). We use 1,000 samples to construct importance sampling estimates, and each sample is generated using 100 cycles of the Gibbs sampler. We also implement $\text{DA-AVM}_{\mathrm{F}}$ based on the MPLE and its associated standard error, which are computed using the {\tt potts} package.

\begin{table}[htbp]
  \centering
    \begin{tabular}{cccccc}
        \toprule
         Method & $\theta$ & Time (min) & \# AV simulations & Eff & ESS/Time\\
        \midrule
         AVM & \multicolumn{1}{c}{\begin{tabular}[c]{@{}c@{}}0.77\\ (0.70, 0.84)\end{tabular}} & 77.5 & 50,000 & - & 7.86\\
         $\text{DA-AVM}_{\mathrm{GP40}}$ & \multicolumn{1}{c}{\begin{tabular}[c]{@{}c@{}}0.77\\ (0.70, 0.84)\end{tabular}} & 50.4 & 29,081 & 0.70&10.09\\
         $\text{DA-AVM}_{\mathrm{F}}$ & \multicolumn{1}{c}{\begin{tabular}[c]{@{}c@{}}0.77\\ (0.70, 0.84)\end{tabular}} & 41.6 & 26,912 & 0.72&13.80\\
        \bottomrule
    \end{tabular}
    \caption{The posterior mean and 95\% HPD interval of $\theta$ for the Potts model on a $32 \times 32$ lattice. The simulated truth of $\theta=0.8$. 50,000 MCMC samples are generated from each method. For the DA-AVM methods, the reported computing times include the construction of the surrogate models.
    }
    \label{potts_table}
\end{table}

\begin{figure}[htbp!]
\centering
{{\includegraphics[width=1\linewidth]{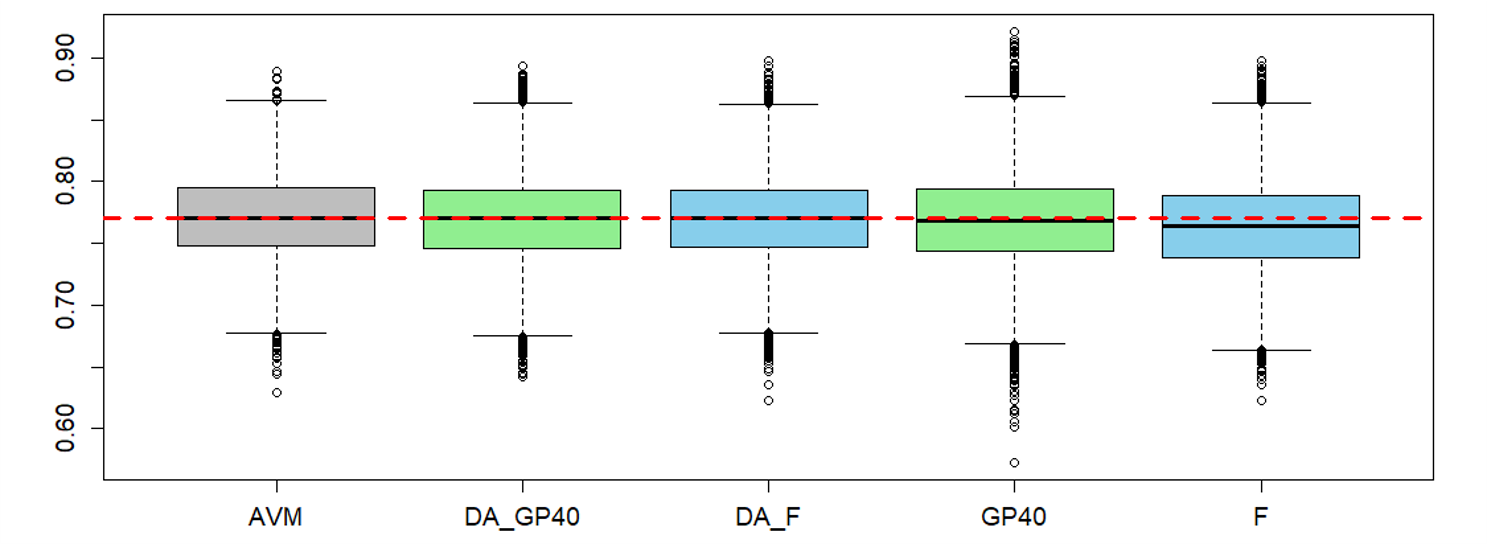} }}%
\caption{Density comparison of $\theta$ for the Potts model. The posterior densities and their corresponding surrogate densities are illustrated in the same color. The red dotted line indicates the posterior median obtained from the AVM. 
}%
\label{potts_box}%
\end{figure}

Table~\ref{potts_table} indicates that the posterior mean estimates from different methods are well aligned to the simulated truth of $\theta=0.8$. Furthermore, we observe that the number of auxiliary variable simulations was reduced by half, resulting in a significant reduction in computing time. For both $\text{DA-AVM}_{\mathrm{GP}}$ and $\text{DA-AVM}_{\mathrm{F}}$, among all rejected proposals, approximately 70 percent of them are filtered in the first stage. This implies that the surrogate models in both approaches are well-constructed and effective. Figure~\ref{potts_box} compares the posterior densities obtained from the DA-AVM methods with their corresponding surrogate densities. The surrogate density is obtained by running a Metropolis-Hastings (MH) algorithm targeting the surrogate distribution, using the same number of iterations as the DA-AVM methods. The similarity between the surrogate and posterior densities indicates that the surrogate models are well constructed. The effective sample size per time (ESS/Time) can be used to assess the efficiency of different algorithms, as it simultaneously accounts for the mixing of the Markov chain and the computational cost. In our study, we also compare ESS/Time across methods. We observe that the ESS/Time values for both DA-AVM methods are higher than those for the standard AVM.

\subsection{An Interaction Point Process Model}

Let $\mathbf x =\lbrace \mathbf x_i \rbrace$ be a realization of spatial point process in a bounded domain $\mathcal W \in \mathbb R^2$. An interaction point process model can describe spatial patterns among points from an interaction function $\phi(d_{ij})$, where $d_{ij}$ is a pairwise distance between $\mathbf x_i$ and $\mathbf x_j$. \cite{attrep} developed a point process to describe the attraction and repulsion patterns of the cells infected with the human respiratory syncytial virus (RSV). The probability model is 
\begin{equation}
\frac{
    \lambda^n \left[ \prod_{i=1}^{n} \exp \left\{ \min \left\{ \sum_{i \neq j} \log (\phi(d_{ij})), 1.2\right\} \right\} \right]
}{
    Z(\bm{\theta})
},~~~\bm \theta =(\lambda, \theta_1,\theta_2,\theta_3).
\label{e27}
\end{equation}
Here, the interaction function is defined as 
\begin{equation}
\phi(d) =
\begin{cases} 
    0, & 0 \leq d \leq R \\
    \theta_1 - \left\{ \frac{\sqrt{\theta_1}}{\theta_2 - R} (d - \theta_2) \right\}^2 & R < d \leq d_1 \\
    1 + \frac{1}{(\theta_3 (d - d_2))^2} & d > d_1.
\end{cases}
\label{e26}
\end{equation}
In the model, $\lambda$ controls the overall intensity of the process; $\theta_1$ represents the maximum value of $\phi$; $\theta_2$ corresponds to the value of $d$ at which $\phi$ attains its maximum; and $\theta_3$ is the decay rate of $\phi$. The calculation of $Z(\bm \theta)$ is intractable because it requires integration in the continuous spatial domain $\mathcal W$. In this example, we analyze the RSV-A point pattern data, consisting of approximately 3,000 points, collected from the $\mathbf{1A2A}$ experiment \citep{attrep}.  Following \cite{attrep}, we use uniform priors for $(\lambda, \theta_1,\theta_2,\theta_3)$ with a range $[2\times 10^{-4},6\times 10^{-4}]\times [1,2]\times[0,20]\times[0,1]$. For all MCMC methods, we run 40,000 iterations, and auxiliary variables are generated using 10 cycles (i.e., $10\times \text{sample size}$) of the birth-death MCMC sampler \citep{geyer1994simulation}.

We implement $\text{DA-AVM}_{\mathrm{GP}}$ using a varying number of particles (ranging from 100 to 400) to cover the 4-dimensional parameter space. Due to the absence of low-dimensional summary statistics in \eqref{e27}, we generate particles using a short run of AVM rather than using the ABC algorithm (Algorithm~\ref{DMHalg} in \ref{algo}). We use 2,000 samples to construct important sampling estimates, and each sample is generated using 10 cycles of the birth-death MCMC sampler. Note that obtaining frequentist estimators for this model is challenging due to the absence of analytical gradients. Instead, we implement $\text{DA-AVM}_{\mathrm{S}}$ using a $1/K$ subsample of the full dataset, where $K = 4, 8, 16$.

\newcolumntype{M}[1]{>{\centering\arraybackslash}m{#1}}
\newcolumntype{P}[1]{>{\centering\arraybackslash}p{#1}}
\begin{table}[htbp]
  \centering
    \begin{tabular}{cccccc}
        \toprule
         Method & $\lambda\times 10^4$ & Time (hr) & \# AV simulations & Eff& ESS/Time\\
        \midrule
         AVM & \multicolumn{1}{c}{\begin{tabular}[c]{@{}c@{}}2.96\\ (2.61, 3.29)\end{tabular}} & 4.56 & 40,000 & - & 11.63\\
         $\text{DA-AVM}_{\mathrm{GP100}}$ & \multicolumn{1}{c}{\begin{tabular}[c]{@{}c@{}}2.97\\ (2.63, 3.27)\end{tabular}} &\multicolumn{1}{c}{\begin{tabular}[c]{@{}c@{}} 3.98\end{tabular}} & \multicolumn{1}{c}{\begin{tabular}[c]{@{}c@{}} 32,112\end{tabular}} & 0.28& 12.20\\
         $\text{DA-AVM}_{\mathrm{GP200}}$ & \multicolumn{1}{c}{\begin{tabular}[c]{@{}c@{}}2.97\\ (2.66, 3.31)\end{tabular}}&\multicolumn{1}{c}{\begin{tabular}[c]{@{}c@{}} 2.28\end{tabular}} & \multicolumn{1}{c}{\begin{tabular}[c]{@{}c@{}} 17,488\end{tabular}} & 0.70&19.02\\
         $\text{DA-AVM}_{\mathrm{GP400}}$ & \multicolumn{1}{c}{\begin{tabular}[c]{@{}c@{}}2.97\\ (2.63, 3.29)\end{tabular}}&\multicolumn{1}{c}{\begin{tabular}[c]{@{}c@{}} 2.58\end{tabular}} & \multicolumn{1}{c}{\begin{tabular}[c]{@{}c@{}} 17,733\end{tabular}} & 0.73& 17.44\\
         $\text{DA-AVM}_{\mathrm{S4}}$ & \multicolumn{1}{c}{\begin{tabular}[c]{@{}c@{}}2.97\\ (2.62, 3.30)\end{tabular}} & 2.21 & 16,707 & 0.65 & 18.60\\
         $\text{DA-AVM}_{\mathrm{S8}}$ & \multicolumn{1}{c}{\begin{tabular}[c]{@{}c@{}}2.99\\ (2.65, 3.32)\end{tabular}} & 2.68 & 22,339 & 0.50 & 16.26\\
         $\text{DA-AVM}_{\mathrm{S16}}$ & \multicolumn{1}{c}{\begin{tabular}[c]{@{}c@{}}2.98\\ (2.64, 3.27)\end{tabular}} & 3.12 & 26,635 & 0.39 &15.32\\
        \bottomrule
    \end{tabular}
    \caption{The posterior mean and 95\% HPD interval of $\lambda\times 10^4$ for an interaction point process model on the RSV-A point pattern data. For the DA-AVM methods, the reported computing times include the construction of the surrogate models.  
    }
    \label{ar_table}
\end{table}

\begin{figure}[htbp!]
\centering
{{\includegraphics[width=1\linewidth]{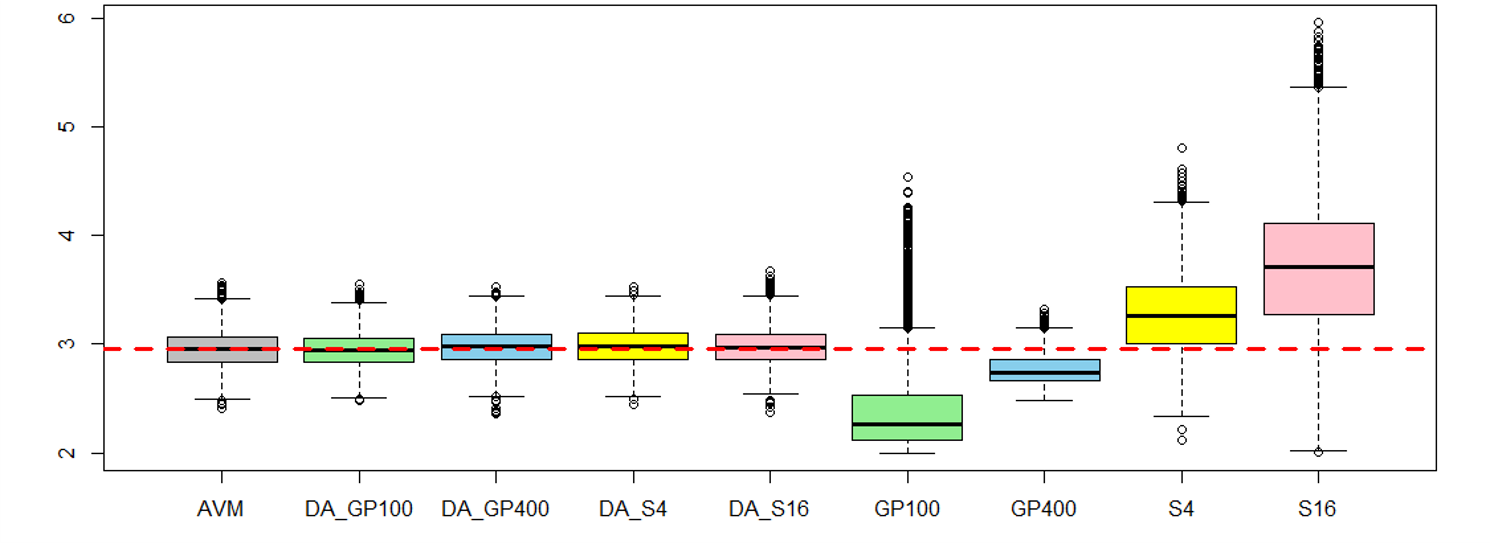} }}%
\caption{Density comparison of $\lambda \times 10^4$ for the interaction point process model. The posterior densities and their corresponding surrogate densities are illustrated in the same color. The red dotted line indicates the posterior median obtained from the AVM.}%
\label{inter_box}%
\end{figure}

Table \ref{ar_table} summarizes the inference results of $\lambda\times 10^4$, indicating that the estimates from the different methods are comparable, while the DA-AVM methods reduce the number of auxiliary variable simulations. Results for the other parameters are provided in \ref{extraresults}. We observe that the efficiency of the DA-AVM methods depends on the surrogate model construction. For $\text{DA-AVM}_\mathrm{GP}$, at least 200 particles are required to achieve 70\% efficiency. Similarly, we observe that the ESS/Time becomes noticeably higher than that of the AVM when at least 200 particles are used. In multidimensional parameter problems, the performance of the Gaussian process emulator is highly dependent on the choice of particles. Consequently, if too few particles are used, purely emulation-based approaches \citep{park2020function, vu2023warped} cannot accurately approximate the target posterior distribution. Similarly, in $\text{DA-AVM}_\mathrm{GP}$, an insufficient number of particles may cause the GP emulator in the first stage kernel to fail to filter proposed candidates efficiently. However, the correction term in the second stage kernel ensures a detailed balance, allowing $\text{DA-AVM}_\mathrm{GP}$ to achieve results comparable to those of the standard AVM. $\text{DA-AVM}_\mathrm{S}$ methods also reduce computational cost because the auxiliary variable simulation in the first stage kernel is low-dimensional compared to the original data. Compared to $\text{DA-AVM}_\mathrm{GP}$, $\text{DA-AVM}_\mathrm{S}$ requires fewer components to be tuned. We observe that $K=4$ is the most efficient in terms of both Eff and ESS/Time. If the subsample size is too small, the resulting surrogate in the first stage kernel becomes flat (i.e., noninformative); therefore, it is likely to accept most proposed candidates. However, $\text{DA-AVM}_\mathrm{S}$ can still approximate the target posterior due to the correction term in the second stage. The importance of surrogate construction is also evident in Figure~\ref{inter_box}. All surrogate densities show deviations from the AVM posterior, with the GP surrogate using 100 particles and the subset-based surrogate with $K=16$ exhibiting particularly large discrepancies, resulting in lower efficiency. However, as discussed earlier, all DA-AVM methods still produce posterior samples comparable to those from AVM.


\subsection{An Exponential Random Graph Model}

Exponential random graph models (ERGMs) \citep{robins2007introduction, hunter2008ergm} are commonly employed to represent social networks as random structures governed by nodal and dyadic interactions. Consider an observed undirected network $\mathbf{x} = \{x_{ij}\}$ with binary adjacency entries $x_{ij} \in \{0,1\}$ for $i < j$, where $x_{ij} = 1$ indicates the presence of an edge between nodes $i$ and $j$, and $x_{ij} = 0$ otherwise. We analyze the Faux Mesa high school network dataset \citep{goodreau2007advances, resnick1997protecting}, which describes a high school friendship network. Each student is associated with covariates such as grade and sex, allowing for the analysis of homophily effects. The corresponding likelihood function is
\begin{align}\label{ergmmodel}
    L(\bm{\theta} | \textbf{x}) &=\frac{\exp\left\{ \bm{\theta}^\top \mathbf{s}(\mathbf{x}) \right\}}{Z(\bm{\theta})},\nonumber\\
    S_1(\textbf{x}) &= \sum_{i=1}^N \binom{x_{i+}}{1},\nonumber \\
    S_{g-5}(\textbf{x}) &= \sum_{i<j} x_{i,j} (1\{ \text{grade}_i = g\} \times 1\{ \text{grade}_j = g\}), \: g = 7,\cdots,12, \nonumber\\
    S_8(\textbf{x}) &= e^{0.25} \sum_{k=1}^{N-1} \left\{1-(1-e^{-0.25})^k\right\} \text{D}_k(\textbf{x}), \nonumber\\
    S_9(\textbf{x}) &= e^{0.25} \sum_{k=1}^{N-2} \left\{1-(1-e^{-0.25})^k\right\} \text{ESP}_k(\textbf{x}),
\end{align}
where $\bm\theta = (\theta_1, \dots, \theta_9)$ are parameters that account for various aspects of the network structure: an edge term for network density, a homophily effect based on grade, a geometrically weighted degree term for degree heterogeneity, and a geometrically weighted edgewise shared partners term for transitivity \citep{snijders2006new}. The computation of $Z(\bm{\theta})$ is intractable, as it requires summation over all $2^{203\choose2}$ possible configurations. Independent normal priors with mean zero and variance 10 are assigned to all parameters. MCMC algorithms are run for 50,000 iterations until convergence, with the first 10,000 samples discarded as burn-in. Auxiliary variables are generated using 10 cycles (i.e., $10 \times 203 \times 203$ iterations) of the Gibbs sampler.

To cover the 9-dimensional parameter space, $\text{DA-AVM}_{\mathrm{GP}}$ is implemented with 400 and 800 particles. The particles are generated using the ABC algorithm described in \ref{algo} (Algorithm~\ref{ABCalg}). Importance sampling estimates are then constructed using 1,000 samples, each of which is generated through 10 cycles of the Gibbs sampler. We also implement $\text{DA-AVM}_{\mathrm{F}}$ based on the MCMLE and its associated observed Fisher information, which are computed using the {\tt ergm} package. We did not implement $\text{DA-AVM}_{\mathrm{S}}$ because partitioning network data while preserving the connectivity structure is not trivial.

\begin{table}[htbp]
  \centering
    \begin{tabular}{cccccc}
        \toprule
         Method & $\theta_1$ & Time (min) & \# AV simulations & Eff & ESS/Time\\
        \midrule
         AVM & \multicolumn{1}{c}{\begin{tabular}[c]{@{}c@{}}-6.35\\ (-6.82, -5.94)\end{tabular}} & 24.4 & 50,000 & - & 162.82 \\
         $\text{DA-AVM}_{\mathrm{GP400}}$ & \multicolumn{1}{c}{\begin{tabular}[c]{@{}c@{}}-6.36\\ (-6.77, -5.90)\end{tabular}} & 17.3 & 34,683 & 0.47 & 229.49\\
         $\text{DA-AVM}_{\mathrm{GP800}}$ & \multicolumn{1}{c}{\begin{tabular}[c]{@{}c@{}}-6.32\\ (-6.72, -5.89)\end{tabular}} & 18.0 & 35,010 & 0.45 & 220.59\\
         $\text{DA-AVM}_{\mathrm{F}}$ & \multicolumn{1}{c}{\begin{tabular}[c]{@{}c@{}}-6.31\\ (-6.70, -5.97)\end{tabular}} & 13.3 & 27,500 & 0.66 & 298.46\\
        \bottomrule
    \end{tabular}
    \caption{The posterior mean and 95\% HPD interval of $\theta_1$ for the ERGM on the Faux Mesa high school network data. For the DA-AVM methods, the reported computing times include the construction of the surrogate models.
    }
    \label{ergm_table}
\end{table}

\begin{figure}[htbp!]
\centering
{{\includegraphics[width=1\linewidth]{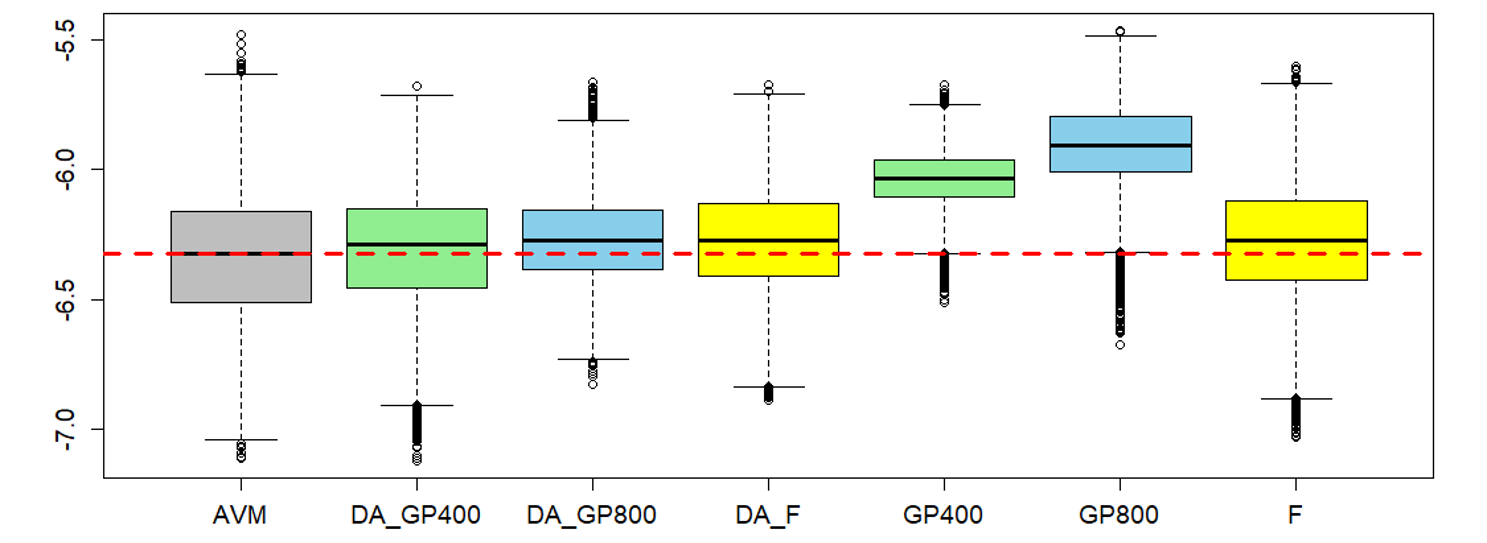} }}%
\caption{Density comparison of $\theta_1$ for the ERGM. The posterior densities and their corresponding surrogate densities are illustrated in the same color. The red dotted line indicates the posterior median obtained from the AVM.}%
\label{ergm1_box}%
\end{figure}

Table~\ref{ergm_table} presents the inference results for $\theta_1$, indicating that the results are comparable across all methods. Results for the other parameters are provided in \ref{extraresults}. As in the previous examples, the DA-AVM methods reduce the number of auxiliary variable simulations compared to the baseline AVM. We observe that $\text{DA-AVM}_{\mathrm{GP}}$ exhibits a relatively low efficiency even with an increasing number of particles. As previously discussed, constructing an accurate Gaussian process emulator is particularly challenging for multidimensional problems. On the other hand, $\text{DA-AVM}_{\mathrm{F}}$ can efficiently filter out proposed candidates because the frequentist estimator-based approximation is more accurate than the Gaussian process emulator for this multidimensional problem. As in previous examples, we compare the posterior densities with their surrogate densities (Figure~\ref{ergm1_box}). The GP surrogate shows noticeable deviation from the target posterior even with increasing particles, while the frequentist-based approximation closely matches the target. This aligns with the higher Eff reported for $\text{DA-AVM}_{\mathrm{F}}$ in Table~\ref{ar_table}. Despite its lower efficiency, $\text{DA-AVM}_{\mathrm{GP}}$ still yields comparable posteriors to AVM due to the second-stage correction. This aspect is further confirmed by the result that $\text{DA-AVM}_{\mathrm{F}}$ achieves a higher ESS/Time compared to the other methods.

\subsection{A Susceptible-Infected-Recovered Model}

Although not a spatial model, this section illustrates the broader applicability of our approach to partially observed Markov processes, which involve intractable likelihoods. Specifically, we study a discrete-time stochastic susceptible-infected-recovered (SIR) model \citep{mckinley2009inference} using historical weekly measles incidence data from Baltimore, available in the {\tt pomp} package \citep{king2016statistical}. Let $S(t), I(t), R(t)$ denote the number of susceptible, infected, and recovered individuals at time $t$, respectively. The latent state evolves according to the following transition dynamics:
\begin{align}\label{SIRlatent}
\Delta N_{SI}(t) &\sim \mathrm{Pois}\left( \frac{\beta S(t) I(t)}{N} \right), \nonumber \\
\Delta N_{IR}(t) &\sim \mathrm{Pois}(\gamma I(t)), \nonumber \\
S({t+1}) &= S(t) - \Delta N_{SI}(t), \nonumber \\
I(t+1) &= I(t) + \Delta N_{SI,t}(t) - \Delta N_{IR}(t), \nonumber \\
R(t+1) &= R(t) + \Delta N_{IR}(t),
\end{align}
where $N$ is the total population, and $\beta$ and $\gamma$ are the transmission and recovery rates. We assume a single infection in the initial state (i.e., $S(0)=N-1, I(0)=1, R(0)=0$). In \eqref{SIRlatent}, $\Delta N_{SI}(t)$ and $\Delta N_{IR}(t)$ denote the number of individuals transitioning from susceptible to infected, and from infected to recovered, respectively, at time $t$. We then introduce a measurement model, where the number of reported infectious cases follows a binomial distribution $\text{case}_t \sim \mathrm{Bin}(I(t), \rho)$, where $\rho$ is the reporting probability. The likelihood function for this model is intractable, as it involves integration over high-dimensional latent state spaces, specifically $\lbrace S(t), I(t), R(t)\rbrace_{t=0}^{T}$. We use independent priors for the model parameters: $\beta \sim \mbox{Log-Normal}(\log 2, 1), \gamma \sim \mbox{Log-Normal} (\log 1, 1), \rho \sim \mbox{Beta}(2,2).$

In this example, we use pseudo-marginal MCMC (PMCMC), where the particle filter \citep{andrieu2010particle} provides an unbiased estimate of the likelihood, which is used in the MH acceptance step. We implement this using the \texttt{pfilter} function from the \texttt{pomp} package. Although the stationary distribution of the PMCMC algorithm is identical to the target posterior distribution, obtaining an unbiased likelihood estimate at each iteration is computationally expensive. To address this, $\text{DA-MCMC}_{\mathrm{F}}$ is implemented using the MLE and observed Fisher information computed via the iterated filtering algorithm \citep{ionides2015inference}, using the \texttt{mif2} function in the \texttt{pomp} package. Once a candidate is accepted in the first stage, an expensive unbiased estimate of the likelihood is computed in the second stage. We use 1,000 particles for both PMCMC and $\text{DA-MCMC}_{\mathrm{F}}$. Each MCMC algorithm is run for 20,000 iterations, with the first 10,000 samples discarded as burn-in.

\begin{table}[htbp]
  \centering
    \begin{tabular}{cccccc}
        \toprule
         Method & $\beta$ & Time (min) & \# AV simulations & Eff & ESS/Time\\
        \midrule
         PMCMC & \multicolumn{1}{c}{\begin{tabular}[c]{@{}c@{}}2.00\\ (1.99, 2.01)\end{tabular}} & 31.32 & 20,000 & - & 31.75 \\
         $\text{DA-MCMC}_{\mathrm{F}}$ & \multicolumn{1}{c}{\begin{tabular}[c]{@{}c@{}}2.00\\ (1.99, 2.01)\end{tabular}} & 8.87 & 7979 & 0.62 & 67.99\\
        \bottomrule
    \end{tabular}
    \caption{The posterior mean and 95\% HPD interval of $\beta$ for the measles data. For the DA-MCMC method, the reported computing times include the construction of the surrogate models.
    }
    \label{pomp_table}
\end{table}

Table~\ref{pomp_table} presents the inference results for $\beta$ in the SIR model. Results for the other parameters are provided in \ref{extraresults}. We observe that the posterior mean estimates from different algorithms are consistent, while $\text{DA-MCMC}_{\mathrm{F}}$ is more efficient in terms of ESS/Time. This improvement arises from the reduced number of expensive unbiased likelihood evaluations. Furthermore, the efficiency score of approximately 0.62 indicates that the surrogate model was well constructed.

\subsection{Guidelines and Recommendations}

We construct the first stage kernel using subsampling, Gaussian process emulation, and frequentist estimators to rule out implausible regions of the parameter space. We observe that variants of DA-AVMs are efficient and accurately approximate the target posterior due to the correction term in the second-stage kernel. Each method has its own advantages and disadvantages, depending on the application. Table~\ref{tab:gp_f_s_tuning} summarizes the DA-AVM methods in terms of the components that require tuning, provides recommendations for tuning these components, and indicates the types of applications for which each method is well suited.

\begin{table}[htbp]
  \centering
    \begin{tabular}{cccc}
        \toprule
         Method & Tune & Recommended & Application \\
        \midrule
         $\text{DA-AVM}_{\mathrm{GP}}$ & $d$, $N$ & $d \approx 100p$, $N = 1,000$--$2,000$ & $p \leq 4$ \\
         $\text{DA-AVM}_{\mathrm{F}}$ & FE & MLE or MPLE & Lattice models, ERGMs \\
         $\text{DA-AVM}_{\mathrm{S}}$ & $K$ & samples per subset $\geq 500$ & No summary stat or FE \\
        \bottomrule
    \end{tabular}
    \caption{Comparison among DA-AVM methods. Here, $d$ denotes the number of particles, $N$ the number of importance samples, FE indicates the frequentist estimators, and $K$ represents the number of subsets.}
    \label{tab:gp_f_s_tuning}
\end{table}

$\text{DA-AVM}_{\mathrm{GP}}$ achieves a high efficiency for low-dimensional parameter problems, although constructing an accurate emulator becomes challenging for multidimensional cases. For $\text{DA-AVM}_{\mathrm{GP}}$, we need to tune the number of particles ($d$) and the number of importance samples ($N$). As the parameter dimension $p$ increases, both $d$ and $N$ should increase accordingly. \cite{park2020function} recommend using $d \approx 100p$ and $N$ between 1,000 and 2,000. In our study, we observed that the surrogate model is well constructed under these settings. $\text{DA-AVM}_{\mathrm{F}}$ performs well when a frequentist estimator is available, such as in the case of network models or spatial lattice models. For these models, we can simply use existing packages (e.g., {\tt ergm} or {\tt potts}) to obtain frequentist estimators. Lastly, $\text{DA-AVM}_{\mathrm{S}}$ can be easily applied without extensive tuning when summary statistics are unavailable or when deriving the analytical gradient of the likelihood is difficult, as in the interaction point process models. For $\text{DA-AVM}_{\mathrm{S}}$, we need to decide the number of subsets ($K$). As mentioned earlier, if the subsample size is too small, the resulting surrogate in the first-stage kernel becomes noninformative. Therefore, we recommend that each subset contain at least 500 samples.

\section{Discussion}
\label{sec:discuss}

In this manuscript, we propose efficient DA-AVM methods that reduce the number of auxiliary variable simulations. We demonstrate that the proposed methods satisfy detailed balance and are ergodic; therefore, DA-AVM algorithms produce samples that converge to the approximate posterior obtained from the AVM. We investigate the application of DA-AVM to a variety of intractable spatial models and show that DA-AVM is computationally more efficient than the standard AVM while providing comparable inference results.

Doubly intractable problems often arise in spatial models where the dimension of the data space $\mathcal{X}$ is much larger than that of the parameter space $\bm{\Theta}$, and this manuscript focuses on studying intractable spatial models. For example, in the ERGM considered in our study, the dimension of $\mathcal{X}$ is $2^{\binom{203}{2}}$, while the dimension of $\bm{\Theta}$ is 9. In this case, the computation of $Z(\bm{\theta})$ requires considering all possible configurations of $\mathcal{X}$, which leads to intractability. The ideas presented here can be broadly applicable to other problems involving intractable likelihoods. One example is models for max-stable processes used in spatial extremes \citep{schlather2002models, kabluchko2009stationary}, where the likelihood function is intractable due to the involvement of high-dimensional derivatives. Another example is state-space models \citep{10.5555/2534502, king2016statistical}, where the latent space is high-dimensional, making the likelihood intractable—even in non-spatial contexts.

From somewhat different perspectives, a variety of computational methods have been developed to address intractable likelihoods. \cite{matsubara2022robust} introduced a robust generalized Bayesian approach that replaces the likelihood with Stein discrepancy-based loss functions, providing theoretical guarantees while avoiding direct likelihood evaluation. More recently, neural network-based approaches have emerged as powerful alternatives. \cite{lenzi2023neural} demonstrated that deep neural networks can be trained on simulated data to directly learn the inverse map from data to parameters for max-stable processes. \cite{walchessen2024neural} constructed neural surrogate models that approximate the likelihood surface itself, enabling efficient emulation of computationally expensive likelihoods in spatial extremes. Building on this, \cite{sainsbury2025neural} extended the framework to irregular spatial domains using graph neural networks. These strategies suggest promising future directions beyond the AVM methods.

Improving the accuracy of surrogate model construction can further enhance the efficiency of the algorithm. For example, \cite{zhou2021markov} employed deep neural networks to approximate computationally expensive forward models within an MCMC framework for inverse problems. In addition, dimension reduction techniques can be considered; for instance, \cite{constantine2014computing, constantine2016accelerating} identify low-dimensional structures in the parameter space to accelerate MCMC sampling in high-dimensional Bayesian inverse problems. A detailed exploration of these methods could further improve the efficiency of DA-AVM, which is an interesting direction for future research.

\section*{Acknowledgement}
This research was supported by the Yonsei University Humanities and Social Sciences Field Creative Research Fund of 2024-22-0581. The authors are grateful to anonymous reviewers for their careful reading and valuable comments.

\clearpage
\appendix

\section{Derivation of the Marginal Transition Kernel of AVM}
\label{deriv_marginal}
For the augmented posterior, $(\bm{\theta}^*, \mathbf{y})$ are updated via a block Gibbs step: first draw $\widetilde{\bm{\theta}^*} \sim q(\widetilde{\bm{\theta}^*}|\bm{\theta})$, then $\widetilde{\mathbf{y}} \sim \frac{h(\widetilde{\mathbf{y}}|\widetilde{\bm{\theta}^*})}{Z(\widetilde{\bm{\theta}^*})}$. Here, $\widetilde{\bm{\theta}^*}$ and $\widetilde{\mathbf{y}}$ denote the updated values of $\bm{\theta}^*$ and $\mathbf{y}$. Then  $\bm{\theta}$ is updated via a Metropolis-Hastings step that exchanges parameter settings. The corresponding Markov transition kernel on the augmented space is
    \begin{align}
K_{\mathrm{AVM}}((\bm{\theta}, \bm{\theta}^*, \mathbf{y}), \mathbf{B}) 
&= \int_{\bm{\Theta}} \int_\mathbf{\mathcal{X}} \underbrace{q(\widetilde{\bm{\theta}^*} | \bm\theta)\, \frac{h(\widetilde{\mathbf{y}}|\widetilde{\bm{\theta}^*})} {Z(\widetilde{\bm{\theta}^*})}}_{\text{block Gibbs sampler}} \Big[\underbrace{  \alpha_{\mathrm{AVM}}(\bm{\theta} \rightarrow \widetilde{\bm{\theta}^*} ; \widetilde{\mathbf{y}})\mathbf{1}_\mathbf{B}( \widetilde{\bm{\theta}^*}, \bm{\theta}, \widetilde{\mathbf{y}})}_{\text{MH - accept}} \nonumber\\
&+ \underbrace{\Big(1-\alpha_{\mathrm{AVM}}(\bm{\theta} \rightarrow \widetilde{\bm{\theta}^*} ; \widetilde{\mathbf{y}})\Big)\mathbf{1}_\mathbf{B}(\bm \theta, \widetilde{\bm{\theta}^*}, \widetilde{\mathbf{y}})}_{\text{MH - reject}}\Big]d\widetilde{\mathbf{y}} \, d\widetilde{\bm{\theta}^*},
\label{kernel_aug_avm}
\end{align}
where $\mathbf{B}\subseteq\mathcal{B}(\bm{\Theta} \times \bm{\Theta} \times \mathcal{X})$ is a measurable set in the Borel $\sigma$-algebra on the product space $\bm{\Theta} \times \bm{\Theta} \times \mathcal{X}$, and the acceptance probability is $$\alpha_{\mathrm{AVM}}(\bm{\theta} \rightarrow \widetilde{\bm{\theta}^*} ; \widetilde{\mathbf{y}})=\min \left\{ 1, 
\frac{p(\widetilde{\bm{\theta}^*}) h(\mathbf{x} | \widetilde{\bm{\theta}^*}) h(\widetilde{\mathbf{y}} | \bm{\theta})  q(\bm{\theta} | \widetilde{\bm{\theta}^*})}
{p(\bm{\theta}) h(\mathbf{x} | \bm{\theta})  h(\widetilde{\mathbf{y}} | \widetilde{\bm{\theta}^*})  q(\widetilde{\bm{\theta}^*} | \bm{\theta})}
\right\}.$$ 
From \eqref{kernel_aug_avm}, the marginal transition kernel targeting $\pi(\bm{\theta}|\mathbf{x})$ is obtained by projection onto the marginal space. Define $\pi_{\bm{\Theta}}(\bm{\theta}, \bm{\theta}^*, \mathbf{y})=\bm{\theta}.$ For $\mathbf{A}\subseteq \mathcal{B}(\bm{\Theta})$, let $\mathbf{B}=\pi_{\bm{\Theta}}^{-1}(\mathbf{A})=\mathbf{A}\times\bm{\Theta}\times\mathcal{X}$. Then the indicators reduce to  $\mathbf{1}_\mathbf{B}( \widetilde{\bm{\theta}^*}, \bm{\theta}, \widetilde{\mathbf{y}})=\mathbf{1}_{\mathbf{A}}(\widetilde{\bm{\theta}^*})$ and $\mathbf{1}_\mathbf{B}(\bm{\theta}, \widetilde{\bm{\theta}^*}, \widetilde{\mathbf{y}})=\mathbf{1}_{\mathbf{A}}(\bm{\theta})$. The induced marginal kernel on $\bm{\Theta}$ is therefore
\begin{align}
    K_{\mathrm{AVM}}(\bm{\theta}, \mathbf{A})&=\int_{\bm{\Theta}} \int_{\mathcal{X}} q(\widetilde{\bm{\theta}^*} | \bm\theta)\, \frac{h(\widetilde{\mathbf{y}}|\widetilde{\bm{\theta}^*})} {Z(\widetilde{\bm{\theta}^*})} \alpha_{\mathrm{AVM}}(\bm{\theta} \rightarrow \widetilde{\bm{\theta}^*} ; \widetilde{\mathbf{y}})\mathbf{1}_{\mathbf{A}}( \widetilde{\bm{\theta}^*})d\widetilde{\mathbf{y}} \, d\widetilde{\bm{\theta}^*} \nonumber \\
    &+ \Big(1-\int_{\bm{\Theta}} \int_{\mathcal{X}} q(\widetilde{\bm{\theta}^*} | \bm\theta)\, \frac{h(\widetilde{\mathbf{y}}|\widetilde{\bm{\theta}^*})} {Z(\widetilde{\bm{\theta}^*})}\alpha_{\mathrm{AVM}}(\bm{\theta} \rightarrow \widetilde{\bm{\theta}^*} ; \widetilde{\mathbf{y}})d\widetilde{\mathbf{y}} \, d\widetilde{\bm{\theta}^*}\Big)\mathbf{1}_{\mathbf{A}}(\bm \theta).
\end{align}
Define the $\mathbf{y}$-averaged acceptance probability $$\overline{\alpha}_{\mathrm{AVM}}(\bm{\theta} \rightarrow \bm{\theta}^*)=\int_{\mathcal{X}} \alpha_{\mathrm{AVM}}(\bm{\theta} \rightarrow \bm{\theta}^* ; \mathbf{y}) \frac{h(\mathbf{y}|\bm{\theta}^*)}{Z(\bm{\theta}^*)} d\mathbf{y},$$ and the overall acceptance probability
$$r_{\mathrm{AVM}}(\bm{\theta})=\int_{\bm{\Theta}} q(\bm{\theta}^*|\bm{\theta})  \overline{\alpha}_{\mathrm{AVM}}(\bm{\theta} \rightarrow \bm{\theta}^* ) d\bm{\theta}^*.$$
Then the marginal transition kernel simplifies to 
\begin{equation}
K_{\mathrm{AVM}}(\bm{\theta}, \mathbf{A})=\int_{\mathbf{A}} q(\bm{\theta}^*|\bm{\theta}) \overline{\alpha}_{\mathrm{AVM}}(\bm{\theta} \rightarrow \bm{\theta}^*) d\bm{\theta}^* + (1-r_{\mathrm{AVM}}(\bm{\theta})) \mathbf{1}_{\mathbf{A}}(\bm{\theta}).    
\label{kernel_marginal_avm}
\end{equation}
Note that \eqref{kernel_marginal_avm} coincides with the expression given in the proof of Corollary 2.3 in the appendix of \cite{alquier2014noisy}.
\clearpage

\section{Algorithm Details} \label{algo}
\begin{algorithm}[hbt]
\caption{$\text{DA-AVM}_{\mathrm{S}}$ algorithm}\label{daavsp}
\small \textbf{Input}: $N$ (the number of MCMC iterations), $m_1$ and $m_2$ (the lengths of the inner sampler in the first and second stages, respectively, where typically $m_1 \ll m_2$)

\small \textbf{Output}: Posterior samples $(\bm\theta^{(1)},\ldots,\bm\theta^{(N)})$ 

\begin{algorithmic}[1]
    \For {$n=0,1,\ldots,N-1$}
    \State \small Sample $\mathbf{x}_{\mathrm{sub}} \in \mathcal{X}_{\mathrm{sub}} \subset \mathcal{X}$
    \State \small $\bm\theta^*\sim q(\cdot|\bm \theta^{(n)})$ with step size $\sigma^{(n)}$
    \State \small $\mathbf y_{\mathrm{sub}}\sim h(\cdot|\bm{\theta}^*)/Z(\bm{\theta}^*)$ by using an $m_1$ iterations of the inner sampler
    \State \small $\alpha_{\mathrm{S1}} \gets \min\left\{1, \frac{p(\bm\theta^*)h(\mathbf x_{\mathrm{sub}}| \bm \theta^*)h(\mathbf y_{\mathrm{sub}}| \bm \theta^{(n)}) q(\bm\theta^{(n)} | \bm\theta^*)}{p(\bm\theta^{(n)})h(\mathbf x_{\mathrm{sub}}|\bm \theta^{(n)})h(\mathbf y_{\mathrm{sub}}| \bm \theta^*) q(\bm\theta^* | \bm\theta^{(n)})} \right\}$

    \State \small $u\sim\text{Unif}[0,1]$
    \If{$u<\alpha_{\mathrm{S1}}$}
        \State $\mathbf y\sim h(\cdot|\bm{\theta}^*)/Z(\bm{\theta}^*)$ by using an $m_2$ iterations of the inner sampler
        \State $\alpha_{\mathrm{S2}} \gets \min \left\{ 1, \frac{ h(\mathbf{x}| \bm{\theta}^*) h(\mathbf y | \bm \theta^{(n)})h(\mathbf x_{\mathrm{sub}}|\bm \theta^{(n)})h(\mathbf y_{\mathrm{sub}}| \bm \theta^*)}{ h(\mathbf x| \bm \theta^{(n)})h(\mathbf y | \bm \theta^*)h(\mathbf x_{\mathrm{sub}}| \bm \theta^*)h(\mathbf y_{\mathrm{sub}}| \bm \theta^{(n)}) }\right\}$
        \State $u\sim\text{Unif}[0,1]$
        \If{$u<\alpha_{\mathrm{S2}}$}
            \State $\bm\theta^{(n+1)} \gets \bm\theta^{*}$
        \Else
            \State $\bm\theta^{(n+1)} \gets \bm\theta^{(n)}$
        \EndIf
    \Else
        \State $\bm\theta^{(n+1)} \gets \bm\theta^{(n)}$
    \EndIf
    \EndFor 
\end{algorithmic}
\label{DASUBAL}
\end{algorithm}
\clearpage

\begin{algorithm}[hbt]
\textbf{Part 1: Construct the Gaussian process emulator}
\caption{$\text{DA-AVM}_{\mathrm{GP}}$ algorithm}\label{daavgp}
\begin{algorithmic}
\normalsize

\State {\it{Step 1.}} Generate $\lbrace \mathbf{x}_l \rbrace_{l=1}^{N}$ from a Markov chain whose stationary distribution is $h(\cdot | \widetilde{\bm\theta}) / Z(\widetilde{\bm\theta})$. 

\State {\it{Step 2.}} Compute an importance sampling estimate \eqref{importanceest} at each particle as $\log \widehat{Z}_{\mathrm{IS}}(\bm\theta^{(i)})$ for $i=1,\cdots,d$. 

\State {\it{Step 3.}} Fitting the Gaussian process model to $\lbrace \bm\theta^{(i)}, \log \widehat{Z}_{\mathrm{IS}}(\bm\theta^{(i)})\rbrace_{i=1}^{d}$ via a maximum likelihood approach.\\
\end{algorithmic}

\textbf{Part2. DA-AVM algorithm with the Gaussian process emulator} \\
\small \textbf{Input}: $N$ (the number of MCMC iterations), $m$ (the length of the inner sampler),  $\widehat \pi_{\mathrm{GP}}(\cdot|\mathbf x)$ (the surrogate posterior with Gaussian process emulation)

\small \textbf{Output}: Posterior samples $(\bm\theta^{(1)},\ldots,\bm\theta^{(N)})$ 

\begin{algorithmic}[1]
    \For {$n=0,1,\ldots,N-1$}
    \State \small $\bm\theta^*\sim q(\cdot|\bm \theta^{(n)})$ 

    \State \small $\alpha_{\mathrm{GP1}} \gets \min \left\{ 1, 
\frac{\widehat{\pi}_{\mathrm{GP}}(\bm{\theta}^*|\mathbf x) q(\bm{\theta}^{(n)}| \bm{\theta}^*)}{\widehat{\pi}_{\mathrm{GP}}(\bm{\theta}^{(n)}|\mathbf x) q(\bm{\theta}^* | \bm{\theta}^{(n)})}
\right\}$

    \State \small $u\sim\text{Unif}[0,1]$
    \If{$u<\alpha_{\mathrm{GP1}}$}
        \State $\mathbf y\sim h(\cdot | \bm\theta^*) / Z(\bm\theta^*)$ using an $m$ iterations of the inner sampler
        \State $\alpha_{\mathrm{GP2}} \gets \min \left\{ 1, 
    \frac{p(\bm{\theta}^*) h(\mathbf x | \bm{\theta}^*)h(\mathbf{y} | \bm{\theta}^{(n)}) \widehat{\pi}_{\mathrm{GP}}(\bm\theta^{(n)}|\mathbf x)}
    {p(\bm{\theta}^{(n)}) h(\mathbf{x} | \bm{\theta}^{(n)}) h(\mathbf y | \bm{\theta}^*) \widehat{\pi}_{\mathrm{GP}}(\bm\theta^{*}|\mathbf x)}
    \right\}$
        \State $u\sim\text{Unif}[0,1]$
        \If{$u<\alpha_{\mathrm{GP2}}$}
            \State $\bm\theta^{(n+1)} \gets \bm\theta^{*}$
        \Else
            \State $\bm\theta^{(n+1)} \gets \bm\theta^{(n)}$
        \EndIf
    \Else
        \State $\bm\theta^{(n+1)} \gets \bm\theta^{(n)}$
    \EndIf
    \EndFor 
\end{algorithmic}
\end{algorithm}
\clearpage

\begin{algorithm}[hbt]
\caption{$\text{DA-AVM}_{\mathrm{F}}$ algorithm}\label{daavn}

\small \textbf{Input}: $N$ (the number of MCMC iterations), $m$ (the length of the inner sampler),  $\widehat \pi_{\mathrm{F}}(\cdot|\mathbf x)$ (the surrogate posterior based on the frequentist estimator)

\small \textbf{Output}: Posterior samples $(\bm\theta^{(1)},\ldots,\bm\theta^{(N)})$ 

\begin{algorithmic}[1]
    \For {$n=0,1,\ldots,N-1$}

    \State \small $\bm\theta^*\sim q(\cdot|\bm \theta^{(n)})$

    \State \small $\alpha_{\mathrm{F1}} \gets \min\left\{1, \frac{\widehat{\pi}_{\mathrm{F}}(\bm\theta^*|\mathbf x) q(\bm\theta^{(n)} | \bm\theta^*)}{\widehat{\pi}_{\mathrm{F}}(\bm\theta^{(n)}|\mathbf x) q(\bm\theta^* | \bm\theta^{(n)})} \right\}$

    \State \small $u\sim\text{Unif}[0,1]$
    \If{$u<\alpha_{\mathrm{F1}}$}
        \State $\mathbf y\sim h(\cdot | \bm\theta^*) / Z(\bm\theta^*)$ using an $m$ iterations of the inner sampler
        \State $\alpha_{\mathrm{F2}} \gets \min \left\{ 1, 
    \frac{p(\bm{\theta}^*) h(\mathbf x | \bm{\theta}^*)h(\mathbf{y} | \bm{\theta}^{(n)}) \widehat{\pi}_{\mathrm{F}}(\bm\theta^{(n)}|\mathbf x)}
    {p(\bm{\theta}^{(n)}) h(\mathbf{x} | \bm{\theta}^{(n)}) h(\mathbf y | \bm{\theta}^*) \widehat{\pi}_{\mathrm{F}}(\bm\theta^{*}|\mathbf x)}
    \right\}$
    
        \State $u\sim\text{Unif}[0,1]$
        \If{$u<\alpha_{\mathrm{F2}}$}
            \State $\bm\theta^{(n+1)} \gets \bm\theta^{*}$
        \Else
            \State $\bm\theta^{(n+1)} \gets \bm\theta^{(n)}$
        \EndIf
    \Else
        \State $\bm\theta^{(n+1)} \gets \bm\theta^{(n)}$
    \EndIf
    \EndFor 
\end{algorithmic}
\end{algorithm}

\clearpage

\begin{algorithm}
\caption{AVM algorithm}\label{DMHalg}
\small \textbf{Input}: $N$ (the number of MCMC iterations), $m$ (the length of inner sampler)  \\
\small \textbf{Output}: Posterior samples $(\bm\theta^{(1)},\ldots,\bm\theta^{(N)})$ 

\begin{algorithmic}[1]
    \For {$n=0,1,\ldots,N-1$}

    \State \small $\bm{\theta}^* \sim q(\cdot|\bm{\theta}^{(n)})$

    \State \small $\mathbf{y} \sim h(\cdot|\bm{\theta}^*)/Z(\bm{\theta}^*)$ using an $m$ iterations of the inner sampler

    \State \small $\alpha \gets \min\left\lbrace 1, \frac{ p(\bm{\theta}^*)h(\mathbf{x}|\bm{\theta}^*)h(\mathbf{y}|\bm{\theta}^{(n)})q(\bm{\theta}^{(n)}|\bm{\theta}^*)}{p(\bm{\theta}^{(n)})h(\mathbf{x}|\bm{\theta}^{(n)})h(\mathbf{y}|\bm{\theta}^*)q(\bm{\theta}^*|\bm{\theta}^{(n)})} \right\rbrace$

    \State \small $u\sim\text{Unif}[0,1]$
    \If{$u<\alpha$}
        \State $\bm\theta^{(n+1)} \gets \bm{\theta}^*$
        \Else
            \State $\bm\theta^{(n+1)} \gets \bm\theta^{(n)}$
    \EndIf
    \EndFor 
\end{algorithmic}
\end{algorithm}

\begin{algorithm}
\caption{ABC algorithm}\label{ABCalg}
\small \textbf{Input}: $(\widehat{\bm\theta},\widehat{\bm\sigma})$ (frequentist estimator and its standard error), $S(\cdot)$ (summary statistics of the model),  $D$ (the number of design points), $\epsilon$ (criterion), $d$ (the number of particles)  \\
\small \textbf{Output}: $d$ number of particles  $(\bm\theta^{(1)},\ldots,\bm\theta^{(d)})$ 
\begin{algorithmic}[1]
    \State $\mathcal{D}_{1}\gets[\widehat{\bm{\theta}} - 10 \widehat{\bm{\sigma}}, \widehat{\bm{\theta}} + 10 \widehat{\bm{\sigma}}]$

    \State $\mathcal{I}=\{\}$
    
    \For {$n=1,\ldots,D$}
    
    $\bm v^{(n)}\sim \text{Unif}[\mathcal{D}_{1}]$ using Latin hypercube design
    
    $\mathbf{y}^{(n)} \sim h(\cdot|\bm v^{(n)})/Z(\bm v^{(n)})$

    \If{$\|S(\mathbf{y}^{(n)}) - S(\mathbf{x})\| < \epsilon$}
        \State Add index $n$ to $\mathcal{I}$
    \EndIf
    \EndFor

    \State $\mathcal{D}_2 \gets [\min_{j\in \mathcal I}\{\bm v^{(j)}\}, \ \max_{j\in \mathcal I}\{\bm v^{(j)}\}]$ where $\mathcal{D}_2 \subset \mathcal{D}_1$

    \For {$n=1,\ldots,d$}

    $\bm\theta^{(n)} \sim \text{Unif}[\mathcal{D}_2]$ using Latin hypercube design

    \EndFor
\end{algorithmic}
\end{algorithm}
\clearpage

\section{Extra Results}\label{extraresults}

\subsection{An Interaction Point Process Model}

\begin{table}[htbp]
  \centering
    \begin{tabular}{M{30mm}M{26mm}M{26mm}M{26mm}} \toprule
    \text{AVM} &\multicolumn{1}{M{26mm}}{${\theta}_1$} &\multicolumn{1}{M{26mm}}{${\theta}_2$} &\multicolumn{1}{M{26mm}}{${\theta}_3$}\\ \midrule
        Posterior mean  & 1.34 & 11.50   &  0.22  \\ 
        95$\%$ HPD  &(1.29, 1.39) & (10.65, 12.27)   &(0.17, 0.27)   \\ \toprule
    \text{$\text{DA-AVM}_{\mathrm{GP100}}$}&\multicolumn{1}{M{26mm}}{${\theta}_1$} &\multicolumn{1}{M{26mm}}{${\theta}_2$} &\multicolumn{1}{M{26mm}}{${\theta}_3$}\\ \midrule
        Posterior mean  &1.34 &11.48 & 0.22  \\ 
        95$\%$ HPD  &(1.30, 1.39) &(10.65, 12.19) &(0.17, 0.27) \\ \toprule
    \text{$\text{DA-AVM}_{\mathrm{GP200}}$} &\multicolumn{1}{M{26mm}}{${\theta}_1$} &\multicolumn{1}{M{26mm}}{${\theta}_2$} &\multicolumn{1}{M{26mm}}{${\theta}_3$}\\ \midrule
        Posterior mean  &1.34 &11.48 & 0.22\\ 
        95$\%$ HPD & (1.30, 1.39)&(10.70, 12.28) &(0.17, 0.27)  \\ \toprule
    \text{$\text{DA-AVM}_{\mathrm{GP400}}$}&\multicolumn{1}{M{26mm}}{${\theta}_1$} &\multicolumn{1}{M{26mm}}{${\theta}_2$} &\multicolumn{1}{M{26mm}}{${\theta}_3$}\\ \midrule
        Posterior mean & 1.34 &11.48 & 0.22\\ 
        95$\%$ HPD  &(1.29, 1.39) & (10.71, 12.30)&(0.17, 0.28) \\ \toprule
    \text{$\text{DA-AVM}_{\mathrm{S4}}$} &\multicolumn{1}{M{26mm}}{${\theta}_1$} &\multicolumn{1}{M{26mm}}{${\theta}_2$} &\multicolumn{1}{M{26mm}}{${\theta}_3$}\\ \midrule
        Posterior mean & 1.34 &11.50 &0.22 \\ 
        95$\%$ HPD &(1.30, 1.39) & (10.60, 12.31) &(0.17, 0.28) \\ \toprule  
    \text{$\text{DA-AVM}_{\mathrm{S8}}$} &\multicolumn{1}{M{26mm}}{${\theta}_1$} &\multicolumn{1}{M{26mm}}{${\theta}_2$} &\multicolumn{1}{M{26mm}}{${\theta}_3$}\\ \midrule
        Posterior mean & 1.34 &11.49 & 0.22\\ 
        95$\%$ HPD  &(1.30, 1.40) &(10.65, 12.31) &(0.18, 0.29) \\ \toprule    
    \text{$\text{DA-AVM}_{\mathrm{S16}}$} &\multicolumn{1}{M{26mm}}{${\theta}_1$} &\multicolumn{1}{M{26mm}}{${\theta}_2$} &\multicolumn{1}{M{26mm}}{${\theta}_3$}\\ \midrule
        Posterior mean & 1.34 & 11.52& 0.22\\ 
        95$\%$ HPD  &(1.29, 1.39) & (10.75, 12.37) & (0.17, 0.27) \\
\bottomrule
    \end{tabular}
    \caption{Posterior inference results for the interaction point process model. The computing times, number of auxiliary variable simulations, and efficiencies are identical to those reported in the main manuscript.
    }
    \label{tab:inter_all}
\end{table}

\subsection{An Exponential Random Graph Model}
\begin{table}[hbt]
\small
  \centering
    \begin{tabular}{M{30mm}M{24mm}M{24mm}M{24mm}M{24mm}} \toprule
    \text{AVM} &\multicolumn{1}{M{24mm}}{${\theta}_2$} &\multicolumn{1}{M{24mm}}{${\theta}_3$} &\multicolumn{1}{M{24mm}}{${\theta}_4$}
    &\multicolumn{1}{M{24mm}}{${\theta}_5$}\\ \midrule
        Posterior mean  & 1.89   &2.08  &1.90 &2.05\\ 
        95$\%$ HPD  & (1.56, 2.18) & (1.75, 2.42) & (1.52, 2.28) & (1.52, 2.59)\\
\cmidrule(l){2-5}
\textbf{} &\multicolumn{1}{M{24mm}}{${\theta}_6$} &\multicolumn{1}{M{24mm}}{${\theta}_7$} &\multicolumn{1}{M{24mm}}{${\theta}_8$} &\multicolumn{1}{M{24mm}}{${\theta}_9$} \\ \cmidrule(l){2-5}
        Posterior mean  & 2.35&2.76 &0.04&1.54  \\ 
        95$\%$ HPD  & (1.98, 2.76)&(2.15, 3.40) &(-0.43, 0.46)&(1.24, 1.81) \\ \toprule
    \text{$\text{DA-AVM}_{\mathrm{GP400}}$} &\multicolumn{1}{M{24mm}}{${\theta}_2$} &\multicolumn{1}{M{24mm}}{${\theta}_3$} &\multicolumn{1}{M{24mm}}{${\theta}_4$}
    &\multicolumn{1}{M{24mm}}{${\theta}_5$}\\ \midrule
        Posterior mean  & 1.86 &2.04 &1.84 &2.04\\ 
        95$\%$ HPD  & (1.59, 2.17) &(1.72, 2.37) &(1.48, 2.18) &(1.48, 2.53)\\
\cmidrule(l){2-5}
\textbf{} &\multicolumn{1}{M{24mm}}{${\theta}_6$} &\multicolumn{1}{M{24mm}}{${\theta}_7$} &\multicolumn{1}{M{24mm}}{${\theta}_8$} &\multicolumn{1}{M{24mm}}{${\theta}_9$} \\ \cmidrule(l){2-5}
        Posterior mean  & 2.34  &2.69 &0.08&1.56 \\ 
        95$\%$ HPD  & (1.95, 2.70) &(2.01, 3.25) & (-0.35, 0.54) & (1.31, 1.84)  \\ \toprule
    \text{$\text{DA-AVM}_{\mathrm{GP800}}$} &\multicolumn{1}{M{24mm}}{${\theta}_2$} &\multicolumn{1}{M{24mm}}{${\theta}_3$} &\multicolumn{1}{M{24mm}}{${\theta}_4$}
    &\multicolumn{1}{M{24mm}}{${\theta}_5$}\\ \midrule
        Posterior mean  &  1.84 & 2.03  & 1.84 &1.99\\ 
        95$\%$ HPD  & (1.52, 2.21) & (1.65, 2.42) & (1.44, 2.29) & (1.39, 2.62) \\
\cmidrule(l){2-5}
\textbf{} &\multicolumn{1}{M{24mm}}{${\theta}_6$} &\multicolumn{1}{M{24mm}}{${\theta}_7$} &\multicolumn{1}{M{24mm}}{${\theta}_8$} &\multicolumn{1}{M{24mm}}{${\theta}_9$} \\ \cmidrule(l){2-5}
        Posterior mean  & 2.32  & 2.58& 0.06 & 1.54 \\ 
        95$\%$ HPD  & (1.90, 2.71) & (1.82, 3.25) & (-0.33, 0.45) & (1.26, 1.84) \\ \toprule
    \text{$\text{DA-AVM}_{\mathrm{F}}$} &\multicolumn{1}{M{24mm}}{${\theta}_2$} &\multicolumn{1}{M{24mm}}{${\theta}_3$} &\multicolumn{1}{M{24mm}}{${\theta}_4$}
    &\multicolumn{1}{M{24mm}}{${\theta}_5$}\\ \midrule
        Posterior mean  &  1.86& 2.04& 1.85&2.02\\ 
        95$\%$ HPD  & (1.53, 2.18) & (1.67, 2.40) & (1.41, 2.21) &(1.52, 2.52)\\
\cmidrule(l){2-5}
\textbf{} &\multicolumn{1}{M{24mm}}{${\theta}_6$} &\multicolumn{1}{M{24mm}}{${\theta}_7$} &\multicolumn{1}{M{24mm}}{${\theta}_8$} &\multicolumn{1}{M{24mm}}{${\theta}_9$} \\ \cmidrule(l){2-5}
        Posterior mean  & 2.35 &2.71 & 0.05&1.53 \\ 
        95$\%$ HPD  & (1.91, 2.73)&(2.05, 3.31) &(-0.30, 0.37) & (1.26, 1.76)  \\ \bottomrule    
    \end{tabular}
    \caption{Posterior inference results for the ERGM on the Faux Mesa high school network data. The computing times, number of auxiliary variable simulations, and efficiencies are identical to those reported in the main manuscript. 
    }
    \label{tab:ergm_all}
\end{table}
\clearpage

\subsection{Susceptible-Infected-Recovered Models}

\begin{table}[htbp]
  \centering
    \begin{tabular}{M{30mm}M{24mm}M{24mm}} \toprule
    \text{PMCMC} &\multicolumn{1}{M{24mm}}{$\gamma$} &\multicolumn{1}{M{24mm}}{$\rho$}\\ \midrule
        Posterior mean  & 0.50 & 0.10 \\ 
        95$\%$ HPD  &(0.50, 0.50) & (0.10, 0.10) \\ \toprule
    \text{$\text{DA-MCMC}_{\mathrm{F}}$}&\multicolumn{1}{M{24mm}}{$\gamma$} &\multicolumn{1}{M{24mm}}{$\rho$}\\ \midrule
        Posterior mean  &0.50 &0.10 \\ 
        95$\%$ HPD  &(0.50, 0.51) &(0.10, 0.10)\\
\bottomrule
    \end{tabular}
    \caption{Posterior inference results for the susceptible-infected-recovered models. The computing times, number of auxiliary variable simulations, and efficiencies are identical to those reported in the main manuscript.
    }
    \label{tab:sir_all}
\end{table}

\clearpage


\bibliographystyle{apalike} 
\bibliography{ref}

\end{document}